\documentclass{article}
\usepackage{amsmath,amsfonts,amssymb,amsthm}
\usepackage{url}
\usepackage{fullpage}

\newtheorem{theorem}{Theorem}
\newtheorem{lemma}{Lemma}
\newtheorem{proposition}{Proposition}
\newtheorem{definition}{Definition}

\title{Short collusion-secure fingerprint codes against three pirates}
\author{Koji Nuida}
\date{Research Center for Information Security (RCIS), National Institute of Advanced Industrial Science and Technology (AIST), 1-18-13 Sotokanda, Chiyoda-ku, Tokyo 101-0021, Japan \\ \url{k.nuida[at]aist.go.jp}}

\begin{document}

\maketitle

\begin{abstract}
In this article, we propose a new construction of probabilistic collusion-secure fingerprint codes against up to three pirates and give a theoretical security evaluation.
Our pirate tracing algorithm combines a scoring method analogous to Tardos codes (J.\ ACM, 2008) with an extension of parent search techniques of some preceding $2$-secure codes.
Numerical examples show that our code lengths are significantly shorter than (about $30\%$ to $40\%$ of) the shortest known $c$-secure codes by Nuida et al.\ (Des.\ Codes Cryptogr., 2009) with $c = 3$.
Some preliminary proposal for improving efficiency of our tracing algorithm is also given.
%\keywords{Collusion-secure code \and Fingerprint code \and 3-secure code \and Content protection}
\end{abstract}

\section{Introduction}
\label{sec:intro}

\subsection{Background and Related Works}
\label{subsec:intro_background}

Recently, digital content distribution services have been widespread by virtue of progress of information technology.
Digitization of content distribution has improved convenience for ordinary people.
However, the digitization also enables malicious persons to perform more powerful attacks, and the amount of illegal content redistribution is increasing very rapidly.
Hence technical countermeasures for such illegal activities are strongly desired.
A use of fingerprint code is a possible solution for such problems, which aims at giving traceability of the attacker (pirate) when an illegally redistributed digital content is found, thus letting the potential attackers abandon to perform actual attacks.

In the context of fingerprint codes, each copy of a content is divided into several segments (common to all copies), in each of which a bit of an encoded user ID is embedded by the content provider by using watermarking technique.
The embedded encoded ID (fingerprint) provides traceability of an adversarial user (pirate) when an unauthorized copy of the content is distributed.
Such a scheme aims at tracing some pirate, without falsely tracing any innocent user, from the fingerprint embedded in the pirated content with an overwhelming probability.
It has been noticed that a coalition of pirates can perform certain strong attacks (collusion attacks) to the fingerprint, therefore any effective fingerprint code should be secure against collusion attacks, called collusion-secure codes.
In particular, if the code is secure against collusion attacks by up to $c$ pirates, then the code is called $c$-secure \cite{BS98}.

Several constructions of collusion-secure codes have been proposed so far.
Among them, the one proposed by Tardos \cite{Tar08} is \lq\lq asymptotically optimal'', in the sense that the order of his code length with respect to the allowable number $c$ of pirates is theoretically the lowest (which is quadratic in $c$).
For improvements of Tardos codes, the constant factor of the asymptotic code length has been reduced by $c$-secure codes given by Nuida et al.\ \cite{NFHKWOI09} to approximately $5.35\%$ of Tardos codes, which is the smallest value so far provable without any additional assumption.
On the other hand, after the first proposal of Tardos codes there were proposed several collusion-secure codes, e.g., \cite{BK04,CFS05,Nui09,NFHIKOW09,NHKWOFI07}, which restrict the number of pirates to $c = 2$ but achieve further short code lengths.
Such constructions of short $c$-secure codes for a small $c$ would have not only theoretical but also practical importance; for example, when the users are less anonymous for the content provider (e.g., the case of secret documents distributed in a company), it seems infeasible to make a large coalition confidentially.
The aim of this article is to extend such a \lq\lq compact'' construction to the next case $c = 3$.

For related works, we notice that there is an earlier work by Seb\'{e} and Domingo-Ferrer \cite{SD02} for $3$-secure codes.
On the other hand, there is another work by Kitagawa et al.\ \cite{KHNWI08} on construction of $3$-secure codes, in which very short code lengths are proposed but its security is evaluated only by computer experiments for some special attack strategies.

\subsection{Our Contribution}
\label{subsec:intro_contribution}

In this article, we propose a new construction of $3$-secure codes and give a theoretical security evaluation.
The codeword generation algorithm is just a bit-wise random sampling, which has been used by many preceding constructions as well.
The novel point of our construction is in the pirate tracing algorithm, which combines the use of score computation analogous to Tardos codes \cite{Tar08} with an extension of \lq\lq parent search'' technique of some preceding works against two pirates \cite{BK04,Nui09,NHKWOFI07}.
Intuitively, the score computation method works well when the parts of fingerprint in the pirated content are chosen evenly from the codewords of pirates, while the extended \lq\lq parent search'' technique works well when the fingerprint is not evenly chosen from the codewords of pirates, therefore their combination is effective.

In comparison under some parameter choices, our code lengths are approximately $3\%$ to $4\%$ of $3$-secure codes by Seb\'{e} and Domingo-Ferrer \cite{SD02}, and approximately $30\%$ to $40\%$ of $c$-secure codes by Nuida et al.\ \cite{NFHKWOI09} for $c = 3$.
This shows that our code length is even significantly shorter than the shortest known $c$-secure codes \cite{NFHKWOI09}.

In fact, Kitagawa et al.\ \cite{KHNWI08} claimed that their $3$-secure code provides almost the same security level as our code for the case of $100$ users and $128$-bit length.
However, they evaluated the security by only computer experiments for the case of some special attack algorithms (and they studied just one parameter choice as above), while in this article we give a theoretical security evaluation for arbitrary attack algorithms under the standard Marking Assumption (cf., \cite{BS98}).
(One may think that the perfect protection of so-called undetectable positions required by Marking Assumption is not practical.
However, this is in fact not a serious problem, as a general conversion technique recently proposed by Nuida \cite{Nui_preprint} can supply robustness against erasure of a bounded number of undetectable bits.)

Moreover, for the sake of improving efficiency of our tracing algorithm, we also discuss an implementation method for the algorithm.
By an intuitive observation, it seems indeed more efficient for an average case than the naive implementation.
A detailed evaluation of the proposed implementation method will be a future research topic.

\subsection{Notations}
\label{subsec:intro_notation}

In this article, $\log$ denotes the natural logarithm.
We put $[n] = \{1,2,\dots,n\}$ for an integer $n$.
Unless some ambiguity emerges, we often abbreviate a set $\{i_1,i_2,\dots,i_k\}$ to $i_1i_2 \cdots i_k$.
Let $\delta_{a,b}$ denote Kronecker delta, i.e., we have $\delta_{a,b} = 1$ if $a = b$ and $\delta_{a,b} = 0$ if $a \neq b$.
For a family $\mathcal{F}$ of sets, let $\bigcup \mathcal{F}$ and $\bigcap \mathcal{F}$ denote the union and the intersection, respectively, of all members of $\mathcal{F}$.

\subsection{Organization of the Article}
\label{subsec:intro_organization}

In Sect.~\ref{sec:preliminary}, we give a formal definition of the notion of collusion-secure fingerprint codes.
In Sect.~\ref{sec:results}, we describe our codeword generation algorithm and pirate tracing algorithm, state the main results on the security of our $3$-secure codes, and give some numerical examples for comparison to preceding works.
Section \ref{sec:proof} summarizes the outline of the security proof.
In Sect.~\ref{sec:implementation_tracing_algorithm}, we discuss an implementation issue of our tracing algorithm.
Finally, Sect.~\ref{sec:proof_detail} supplies the detail of our security proof omitted in Sect.~\ref{sec:proof}.

\section{Collusion-Secure Fingerprint Codes}
\label{sec:preliminary}

In this section, we introduce formal definitions for fingerprint codes.
Let $N$ and $m$ be positive integers, and $1 \leq c \leq N$ an integer parameter.
Put $U = [N]$.
Fix a symbol `$?$' different from `$0$' and `$1$'.
We start with the following definition:
\begin{definition}
\label{defn:fingerprint_game}
Given the parameters $N$, $m$ and $c$, we define the following game, which we refer to as \emph{pirate tracing game}.
The players of the game is a \emph{provider} and \emph{pirates}, and the game is proceeded as follows:
\begin{enumerate}
\item \emph{Provider} generates an $N \times m$ binary matrix $W = (w_{i,j})_{i \in [N], j \in [m]}$ and an element $\mathsf{st}$ called \emph{state information}.
\item \emph{Pirates} generate $U_{\mathrm{P}} \subseteq U$, $1 \leq |U_{\mathrm{P}}| \leq c$, without knowing $W$ and $\mathsf{st}$.
\item \emph{Pirates} receive the codeword $w_i = (w_{i,1},\dots,w_{i,m})$ for every $i \in U_{\mathrm{P}}$.
\item \emph{Pirates} generate a word $y = (y_1,\dots,y_m)$ on $\{0,1,?\}$ under a certain restriction specified below, and send $y$ to \emph{provider}.
\item \emph{Provider} generates $\mathsf{Acc} \subseteq U$ from $y$, $W$, and $\mathsf{st}$, without knowing $U_{\mathrm{P}}$.
\item Then \emph{pirates} win if $\mathsf{Acc} \cap U_{\mathrm{P}} = \emptyset$ or $\mathsf{Acc} \not\subseteq U_{\mathrm{P}}$, and otherwise \emph{provider} wins.
\end{enumerate}
\end{definition}
We call the word $y$ in Step 4 an \emph{attack word} and call `$?$' an \emph{erasure symbol}.
Put $U_{\mathrm{I}} = U \setminus U_{\mathrm{P}}$.
In the definition, $U$ signifies the set of all users, $U_{\mathrm{P}}$ is the coalition of pirates, and $U_{\mathrm{I}}$ is the set of innocent users.
The codeword $w_i$ signifies the fingerprint for user $i$, and the word $y$ signifies the fingerprint embedded in the pirated content.
The set $\mathsf{Acc}$ consists of the users traced by the provider from the pirated content.
The events $\mathsf{Acc} \cap U_{\mathrm{P}} = \emptyset$ and $\mathsf{Acc} \not\subseteq U_{\mathrm{P}}$ specified in Step 6 are referred to as \emph{false-negative} and \emph{false-positive} (or \emph{false-alarm}), respectively.
Both of false-negative and false-positive are called \emph{tracing error}.

Let $\mathsf{Gen}$, $\mathsf{Reg}$, $\rho$, and $\mathsf{Tr}$ denote the algorithms used in Steps 1, 2, 4, and 5, respectively.
We call $\mathsf{Gen}$, $\mathsf{Reg}$, $\rho$, and $\mathsf{Tr}$ \emph{codeword generation algorithm}, \emph{registration algorithm}, \emph{pirate strategy}, and \emph{tracing algorithm}, respectively.
We refer to the pair $\mathcal{C} = (\mathsf{Gen},\mathsf{Tr})$ as a \emph{fingerprint code}, and the following quantity
\begin{equation}
\begin{split}
Pr[ (W,\mathsf{st}) \leftarrow \mathsf{Gen}();\,
U_{\mathrm{P}} \leftarrow \mathsf{Reg}();\,
y \leftarrow \rho(U_{\mathrm{P}},(w_i)_{i \in U_{\mathrm{P}}}); \\
\mathsf{Acc} \leftarrow \mathsf{Tr}(y,W,\mathsf{st}):\,
\mathsf{Acc} \cap U_{\mathrm{P}} = \emptyset \mbox{ or } \mathsf{Acc} \not\subseteq U_{\mathrm{P}}]
\end{split}
\end{equation}
(i.e., the overall probability that pirates win) is called an \emph{error probability} of $\mathcal{C}$.

We specify the restriction for $y$ mentioned in Step 4.
First we present some terminology.
For $j \in [m]$, $j$-th column in codewords is called \emph{undetectable} if $j$-th bits $w_{i,j}$ of the codewords $w_i$ of pirates $i \in U_{\mathrm{P}}$ coincide with each other; otherwise the column is called \emph{detectable}.
Then, in this article, we put the following standard assumption called \emph{Marking Assumption} \cite{BS98}:
\begin{definition}
\label{defn:MA}
The \emph{Marking Assumption} states the following: For the attack word $y$, for every undetectable column $j$, we have $y_j = w_{i,j}$ for some (or equivalently, all) $i \in U_{\mathrm{P}}$.
\end{definition}

We say that a fingerprint code $\mathcal{C}$ is \emph{collusion-secure} if the error probability of $\mathcal{C}$ is sufficiently small for any $\mathsf{Reg}$ and $\rho$ under Marking Assumption.
More precisely, we say that $\mathcal{C}$ is \emph{$c$-secure} (\emph{with $\varepsilon$-error}) \cite{BS98} if the error probability is not higher than a sufficiently small value $\varepsilon$ under Marking Assumption.

\section{Our $3$-Secure Codes}
\label{sec:results}

Here we propose a codeword generation algorithm $\mathsf{Gen}$ and a tracing algorithm $\mathsf{Tr}$ for $3$-secure codes ($c = 3$).
The security property will be discussed below.

The algorithm $\mathsf{Gen}$, with parameter $1/2 \leq p < 1$, is the codeword generation algorithm of Tardos codes \cite{Tar08} but the probability distribution of biases is different: For each (say, $j$-th) column, each user's bit $w_{i,j}$ is independently chosen by $Pr[ w_{i,j} = 1 ] = p_j$, where $p_j = p$ or $1 - p$ with probability $1/2$ each.
Then $\mathsf{Gen}$ outputs $W = ({w_{i,j})_{i \in [N}, j \in [m]}$ and $\mathsf{st} = (p_j)_{j \in [m]}$.

To describe the algorithm $\mathsf{Tr}$, we introduce some notations.
For binary words $w^{(1)},\dots,\allowbreak w^{(k)}$ of length $m$, we define
\begin{equation}
\mathcal{E}(w^{(1)},\dots,w^{(k)}) = \{y \in \{0,1\}^m \mid y_j \in \{w^{(1)}_j,\dots,w^{(k)}_j\} \mbox{ for every } j \in [m]\} \enspace,
\end{equation}
the \emph{envelope} of $w^{(1)},\dots,w^{(k)}$.
Then for a binary word $y$ of length $m$ and a collection $W = (w_{i,j})$ of codewords of users, we define
\begin{equation}
\label{eq:definition_parent_triple}
\mathcal{T}(y) = \{i_1i_2i_3 \subseteq U \mid i_1 \neq i_2 \neq i_3 \neq i_1, \,y \in \mathcal{E}(w_{i_1},w_{i_2},w_{i_3}) \}
\end{equation}
(see Sect.~\ref{subsec:intro_notation} for the notation $i_1i_2i_3$).
A key property implied by Marking Assumption is that if the attack word $y$ contains no erasure symbols, then $y$ belongs to the envelope of the codewords of pirates and, if furthermore $|U_{\mathrm{P}}| = 3$, the family $\mathcal{T}(y)$ contains the set of three pirates.
By using these notations, we define the algorithm $\mathsf{Tr}$ as follows, where the words $y$, $w_1,\dots,w_N$ and the state information $\mathsf{st} = (p_j)_{j \in [m]}$ are given:
\begin{enumerate}
\item \label{item:tracing_replace_?}
Replace each erasure symbol \lq $?$' in $y$ with \lq $0$' or \lq $1$' independently in the following manner.
If $y_j = {?}$, then it is replaced with \lq $1$' with probability $p_j$, and with \lq $0$' with probability $1-p_j$.
Let $y'$ denote the resulting word.
\item \label{item:tracing_threshold}
Calculate a threshold parameter $Z = Z_{y'}$ as specified below.
\item \label{item:tracing_score_compute}
For each $i \in U$, calculate the score $S(i)$ of $i$ by
\begin{equation}
\label{eq:scoring_function}
S(i) = \sum_{\substack{j \in [m] \\ y'_j = 1}} \delta_{w_{i,j},y'_j} \log\frac{1}{p_j} + \sum_{\substack{j \in [m] \\ y'_j = 0}} \delta_{w_{i,j},y'_j} \log\frac{1}{1-p_j} \enspace.
\end{equation}
\item \label{item:tracing_score_check}
If $S(i) \geq Z$ for some $i \in U$, then output every $i \in U$ such that $S(i) \geq Z$, and halt.
\item \label{item:tracing_triple_compute}
Calculate $\mathcal{T}' = \{T \in \mathcal{T}(y') \mid T \cap T' \neq \emptyset \mbox{ for every }\allowbreak T' \in \mathcal{T}(y')\}$.
If $\mathcal{T}' = \emptyset$, then output nobody, and halt.
\item \label{item:tracing_triple_intersection}
If $\bigcap \mathcal{T}' \neq \emptyset$, then output every member of $\bigcap \mathcal{T}'$, and halt.
\item \label{item:tracing_pair_compute}
Calculate $\mathcal{P} = \{P = i_1i_2 \subseteq U \mid i_1 \neq i_2, P \cap T \neq \emptyset \mbox{ for every } T \in \mathcal{T}'\}$.
Let $\mathcal{P}_k$ be the set of all $i \in U$ such that $|\{P \in \mathcal{P} \mid i \in P\}| = k$.
\item \label{item:tracing_pair_multiplicity_1}
If $\mathcal{P}_1 \neq \emptyset$, then output every $i \in U$ such that $i i' \in \mathcal{P}$ for some $i' \in \mathcal{P}_1$, and halt.
\item \label{item:tracing_pair_7}
If $|\mathcal{P}| = 7$, then output every $i \in U$ such that $i i' \in \mathcal{P}$ for some $i' \in \mathcal{P}_2$, and halt.
\item \label{item:tracing_pair_6}
If $|\mathcal{P}| = 6$, then output every $i \in \mathcal{P}_3$, and halt.
\item \label{item:tracing_pair_5_nonempty}
If $|\mathcal{P}| = 5$ and $\mathcal{T}'' = \{i_1i_2i_3 \in \mathcal{T}' \mid i_1i_2,i_2i_3,i_1i_3 \in \mathcal{P}\} \neq \emptyset$, then output every member of $\mathcal{P}_2 \cap (\bigcup \mathcal{T}'')$, and halt.
\item \label{item:tracing_pair_5_empty}
If $|\mathcal{P}| = 5$ and $\mathcal{T}'' = \emptyset$, then output every $i \in \bigcup \mathcal{P}$ such that $i i' \not\in \mathcal{P}$ for some $i' \in \bigcup \mathcal{P}$, and halt.
\item \label{item:tracing_pair_4}
If $|\mathcal{P}| = 4$, then output every $i \in \bigcup \mathcal{P}$ such that $T \in \mathcal{T}'$ and $T \subseteq \bigcup \mathcal{P}$ imply $i \in T$, and halt.
\item \label{item:tracing_pair_3}
If $|\mathcal{P}| = 3$, then output every $i \in \bigcup \mathcal{P}$, and halt.
\item \label{item:tracing_pair_otherwise}
Output nobody, and halt.
\end{enumerate}
This algorithm is divided into two parts; Steps \ref{item:tracing_replace_?}--\ref{item:tracing_score_check} and the remaining steps.
The former part aims at performing coarse tracing to defy \lq\lq unbalanced'' pirate strategies; namely, if some pirates' codewords contribute to generate $y$ at too many columns than the other pirates, then it is very likely that scores of such pirates exceed the threshold and they are correctly accused by Step \ref{item:tracing_score_check}.

On the other hand, the latter complicated part aims at performing more refined tracing.
First, the algorithm enumerates the collections of three users such that $y'$ can be made (under Marking Assumption) from their codewords, in other words, the collection is a candidate of the actual triple of pirates.
Steps \ref{item:tracing_triple_compute} and \ref{item:tracing_triple_intersection} are designed according to an intuition that a pirate would be very likely to be contained in much more candidate triples than an innocent user.
When the tracing algorithm did not halt until Step \ref{item:tracing_triple_intersection}, the possibilities of \lq\lq structures'' of the set $\mathcal{T}'$ are mostly limited, even allowing us to enumerate all the possibilities.
However, it is space-consuming to enumerate them and determine suitable outputs in a case-by-case manner.
Instead, we give an explicit algorithm (Steps \ref{item:tracing_pair_compute}--\ref{item:tracing_pair_otherwise}) to determine a suitable output, which is artificial but not too space-consuming.
Some examples of the possibilities of $\mathcal{T}'$ are given in Fig.\ \ref{fig:example_latter_part}, where $1$, $2$, $3$ are the pirates, $i_j$ are innocent users and the members of $\mathcal{T}'$ are denoted by triangles.
\begin{figure}[htb]
\centering
\begin{picture}(100,150)(15,-160)
\put(60,-30){\circle{12}}\put(60,-32){\hbox to0pt{\hss$1$\hss}}
\put(40,-60){\circle{12}}\put(40,-62){\hbox to0pt{\hss$2$\hss}}
\put(80,-60){\circle{12}}\put(80,-62){\hbox to0pt{\hss$3$\hss}}
\put(20,-30){\circle{12}}\put(20,-32){\hbox to0pt{\hss$i_1$\hss}}
\put(100,-30){\circle{12}}\put(100,-32){\hbox to0pt{\hss$i_2$\hss}}
\put(40,-90){\circle{12}}\put(40,-92){\hbox to0pt{\hss$i_3$\hss}}
\put(80,-90){\circle{12}}\put(80,-92){\hbox to0pt{\hss$i_4$\hss}}
\put(26,-30){\line(1,0){28}}\put(26,-30){\line(3,-5){14}}\put(54,-30){\line(-3,-5){14}}
\put(66,-30){\line(1,0){28}}\put(66,-30){\line(3,-5){14}}\put(94,-30){\line(-3,-5){14}}
\put(46,-60){\line(1,0){28}}\put(46,-60){\line(3,5){14}}\put(74,-60){\line(-3,5){14}}
\put(41,-66){\line(1,0){38}}\put(41,-66){\line(0,-1){19}}\put(79,-66){\line(-2,-1){38}}
\put(45,-63){\line(1,0){29}}\put(45,-63){\line(1,-1){29}}\put(74,-63){\line(0,-1){29}}
\put(10,-110){$\mathcal{P} = \{12,13,23,3i_1,2i_2\}$}
\put(10,-125){$\mathcal{P}_1 = \{i_1,i_2\}$}
\put(10,-140){output = $2,3$}
\end{picture}
\begin{picture}(15,130)(0,-140)
\multiput(0,0)(0,-20){7}{\line(0,-1){10}}
\end{picture}
\begin{picture}(100,80)(10,-160)
\put(20,-30){\circle{12}}\put(20,-32){\hbox to0pt{\hss$1$\hss}}
\put(60,-30){\circle{12}}\put(60,-32){\hbox to0pt{\hss$3$\hss}}
\put(20,-60){\circle{12}}\put(20,-62){\hbox to0pt{\hss$2$\hss}}
\put(60,-60){\circle{12}}\put(60,-62){\hbox to0pt{\hss$i_1$\hss}}
\put(100,-60){\circle{12}}\put(100,-62){\hbox to0pt{\hss$i_2$\hss}}
\put(26,-30){\line(1,0){28}}\put(26,-30){\line(0,-1){28}}\put(54,-30){\line(-1,-1){28}}
\put(21,-54){\line(1,0){38}}\put(21,-54){\line(0,1){19}}\put(59,-54){\line(-2,1){38}}
\put(66,-60){\line(1,0){28}}\put(66,-60){\line(0,1){28}}\put(94,-60){\line(-1,1){28}}
\put(0,-95){$\mathcal{P} =$}
\put(0,-110){$\{13,1i_1,1i_2,23,2i_1,2i_2,3i_1\}$}
\put(0,-125){$\mathcal{P}_1 = \emptyset$, $\mathcal{P}_2 = \{i_2\}$}
\put(0,-140){output = $1,2$}
\end{picture}
\begin{picture}(20,130)(0,-140)
\multiput(5,0)(0,-20){7}{\line(0,-1){10}}
\end{picture}
\begin{picture}(100,80)(20,-160)
\put(20,-30){\circle{12}}\put(20,-32){\hbox to0pt{\hss$1$\hss}}
\put(60,-30){\circle{12}}\put(60,-32){\hbox to0pt{\hss$3$\hss}}
\put(100,-30){\circle{12}}\put(100,-32){\hbox to0pt{\hss$i_2$\hss}}
\put(20,-60){\circle{12}}\put(20,-62){\hbox to0pt{\hss$2$\hss}}
\put(60,-60){\circle{12}}\put(60,-62){\hbox to0pt{\hss$i_1$\hss}}
\put(100,-60){\circle{12}}\put(100,-62){\hbox to0pt{\hss$i_3$\hss}}
\put(26,-30){\line(1,0){28}}\put(26,-30){\line(0,-1){28}}\put(54,-30){\line(-1,-1){28}}
\put(21,-54){\line(1,0){38}}\put(21,-54){\line(0,1){19}}\put(59,-54){\line(-2,1){38}}
\put(66,-60){\line(1,0){28}}\put(66,-60){\line(0,1){28}}\put(94,-60){\line(-1,1){28}}
\put(61,-36){\line(1,0){38}}\put(61,-36){\line(0,-1){19}}\put(100,-36){\line(-2,-1){38}}
\put(10,-95){$\mathcal{P} = \{13,1i_1,23,2i_1,3i_1\}$}
\put(10,-110){$\mathcal{P}_1 = \emptyset$, $\mathcal{T}'' = \emptyset$}
\put(10,-125){$\bigcup\mathcal{P} = \{1,2,3,i_1\}$, $12 \not\in \mathcal{P}$}
\put(10,-140){output = $1,2$}
\end{picture}
\caption{Examples of the sets $\mathcal{T}'$ and $\mathcal{P}$}
\label{fig:example_latter_part}
\end{figure}
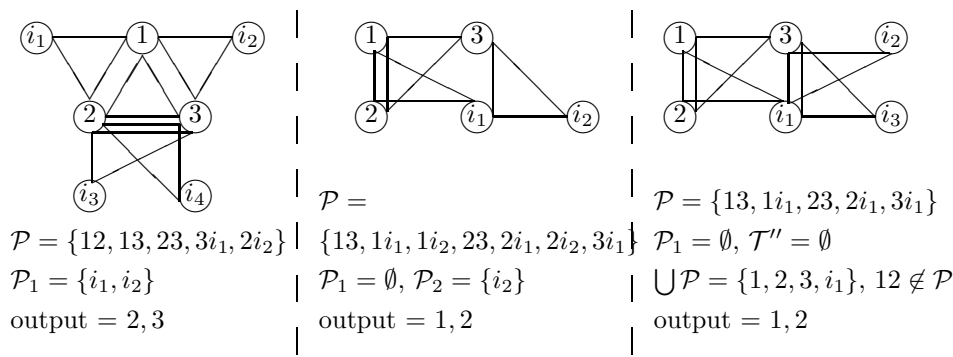

For the latter part, the tracing tends to fail in the case that the set $\mathcal{T}(y')$ contains much more members other than the triple of the pirates, which tends to occur when the contributions of the pirates' codewords to $y$ was too unbalanced.
However, such an unbalanced attack is defied by the former part, therefore the latter part also works well.
More precisely, an upper bound of the error probability at the latter part will be derived by using the property that scores of pirates are lower than the threshold (as otherwise the tracing halts at the former part); cf., Sect.~\ref{subsec:proof_probability_type_IV}.
Our scoring function (\ref{eq:scoring_function}), which is different from the ones for Tardos codes \cite{Tar08} and its symmetrized version \cite{SKC08}, is adopted to simplify the derivation process.
Although it is possible that the true error probability is reduced by applying the preceding scoring functions, a proof of a bound of error probability with those scoring functions requires another evaluation technique and would be much more involved, which is a future research topic.

Note that, for the case $p = 1/2$, it is known that the \lq\lq minority vote'' by three pirates for generating $y$ cancels the mutual information between $y$ and a single codeword, therefore the pirates are likely to escape from the former part of $\mathsf{Tr}$.
However, even by such a strategy the pirates are unlikely to escape from the latter part of $\mathsf{Tr}$, as \emph{collections} of users rather than individual users are considered there.

The threshold parameter $Z = Z_{y'}$ in Step \ref{item:tracing_threshold} is determined as follows.
Let $A_{\mathrm{H}}$ be the set of column indices $j$ such that $(p_j,y'_j) = (p,1)$ or $(1-p,0)$, i.e., the occurrence probability of the bit $y'_j \in \{0,1\}$ at $j$-th column is $p \geq 1/2$, and let $A_{\mathrm{L}} = [m] \setminus A_{\mathrm{H}}$.
Put $a_{\mathrm{H}} = |A_{\mathrm{H}}|$ and $a_{\mathrm{L}} = |A_{\mathrm{L}}|$.
Choose a parameter $\varepsilon_0 > 0$ which is smaller than the desired bound $\varepsilon$ of error probability.
Then choose $Z = Z_{y'}$ satisfying the following condition:
\begin{equation}
\label{eq:threshold}
\sum_{k_{\mathrm{H}},k_{\mathrm{L}}} \binom{a_{\mathrm{L}}}{k_{\mathrm{L}}} p^{a_{\mathrm{L}} - k_{\mathrm{L}}} (1-p)^{k_{\mathrm{L}}} \binom{a_{\mathrm{H}}}{k_{\mathrm{H}}} p^{k_{\mathrm{H}}} (1-p)^{a_{\mathrm{H}} - k_{\mathrm{H}}} \leq \frac{ \varepsilon_0 }{ N } \enspace,
\end{equation}
where the sum runs over all integers $k_{\mathrm{H}},k_{\mathrm{L}} \geq 0$ such that $k_{\mathrm{H}} \log\frac{1}{\,p\,} + k_{\mathrm{L}} \log\frac{1}{1-p} \geq Z$.
An example of a concrete choice of $Z$ satisfying the condition (\ref{eq:threshold}) is as follows:
\begin{equation}
\label{eq:threshold_example}
Z_0 = a_{\mathrm{H}} p \log\frac{1}{\,p\,} + a_{\mathrm{L}} (1-p)\log\frac{1}{1-p}
+ \sqrt{ \frac{1}{\,2\,} \left( \left( \log\frac{1}{\,p\,} \right)^2 a_{\mathrm{H}} + \left( \log\frac{1}{1-p} \right)^2 a_{\mathrm{L}} \right) \log \frac{N}{\varepsilon_0} }
\end{equation}
(see Sect.~\ref{subsec:proof_threshold} for the proof).
From now, we suppose that the threshold $Z$ satisfies the condition (\ref{eq:threshold}) and $Z \leq Z_0$.

For the security of the proposed fingerprint code, first we present the following result, which will be proven in Sect.~\ref{sec:proof}:
\begin{theorem}
\label{thm:error_probability}
By the above choice of $\varepsilon_0$ and $Z$, if the number of pirates is three, then the error probability of the proposed fingerprint code is lower than
\begin{equation}
\varepsilon_0
+ \binom{N-3}{3} f_1(p)^m
+ 3(N-3)(N-4) f_2(p)^m
+ (N-3) (1-p)^{-3 \sqrt{ (m/2) \log(N/\varepsilon_0) } } f_3(p)^m \enspace,
\end{equation}
where we put
\begin{equation}
\label{eq:thm_error_probability_statement_notation}
\begin{split}
f_1(p) &= 1 - 3p^2 + 10p^3 - 15p^4 + 12p^5 - 4p^6 \enspace, \\
f_2(p) &= p^2(1-p)^2(\sqrt{p} + \sqrt{1-p}) + 1 - p - p^2 + 4p^3 - 2p^4 \enspace, \\
f_3(p) &= p^{4-3p}(p^2 - 3p + 3) + (1-p)^{3p+1}(p^2 + p + 1) \enspace.
\end{split}
\end{equation}
\end{theorem}
Some numerical analysis suggests that the choice $p = 1/2$ would be optimal (or at least pretty good) to decrease the bound of error probabilities specified in Theorem \ref{thm:error_probability}.
In fact, an elementary analysis shows that the second term $\binom{N-3}{3} f_1(p)^m$ in the sum, which seems dominant (cf., Theorem \ref{thm:error_probability_p_1/2} below), takes the minimum over $p \in [1/2,1)$ at $p = 1/2$.
Hence we use $p = 1/2$ in the following argument.
Now it is shown that the error probability against less than three pirates also has the same bound under a condition (\ref{eq:error_probability_condition_m}) below (which seems trivial in practical situations), therefore we have the following (which will be proven in Sect.~\ref{sec:proof}):
\begin{theorem}
\label{thm:error_probability_p_1/2}
By using the value $p = 1/2$, the proposed fingerprint code is $3$-secure with error probability lower than
\begin{equation}
\varepsilon_0
+ \binom{N-3}{3} \left( \frac{7}{\,8\,} \right)^m
+ 3(N-3)(N-4) \left( \frac{ 10 + \sqrt{2} }{ 16 } \right)^m
+ (N-3) 8^{\sqrt{ (m/2) \log(N/\varepsilon_0) } } \left( \frac{ 7 \sqrt{2} }{ 16 } \right)^m
\end{equation}
provided
\begin{equation}
\label{eq:error_probability_condition_m}
m \geq 8 \log \frac{N}{\varepsilon_0} \left( 1 + \frac{ 1 }{ 16 \log(N/\varepsilon_0) } \right)^2 \enspace.
\end{equation}
\end{theorem}
Note that when $p = 1/2$, the score $S(i)$ of a user $i$ is equal to $\log 2$ times the number of columns in which the words $w_i$ and $y'$ coincide.
Hence the calculation of scores can be made easier by using the \lq\lq normalized'' score $\widetilde{S}(i) = S(i) / \log 2$ instead, which is equal to $m$ minus the Hamming distance of $w_i$ from $y'$, together with the \lq\lq normalized'' threshold $Z_0 / \log 2 = m/2 + \sqrt{(m/2) \log(N/\varepsilon_0)}$.

Table \ref{tab:comparison_SD02} shows comparison of our code lengths (numerically calculated by using Theorem \ref{thm:error_probability_p_1/2}) with $3$-secure codes by Seb\'{e} and Domingo-Ferrer \cite{SD02}.
Table \ref{tab:comparison_Nui09} shows the comparison with $c$-secure codes by Nuida et al.\ \cite{NFHKWOI09} for $c = 3$.
The values of $N$ and $\varepsilon$ and the corresponding code lengths are chosen from those articles.
The tables show that our code lengths are much shorter than the codes in \cite{SD02}, and even significantly shorter than the codes in \cite{NFHKWOI09} which are in fact the shortest $c$-secure codes known so far (improving the celebrated Tardos codes \cite{Tar08}).
On the other hand, recently Kitagawa et al.\ \cite{KHNWI08} proposed another construction of $3$-secure codes, and evaluated the security against some typical pirate strategies in the case $N = 100$ and $m = 128$ by computer experiment.
The resulting error probability was $\varepsilon = 0.009$.
For the same error probability, our code length (with parameter $\varepsilon_0 = \varepsilon/2$) is $m = 135$.
Therefore our code, which is \emph{provably secure} in contrast to their code, has almost the same length as their code.
\begin{table}[htb]
\centering
\caption{Comparison of code lengths with the codes by Seb\'{e} and Domingo-Ferrer \cite{SD02}}
\label{tab:comparison_SD02}
\begin{tabular}{c||c|c|c}
$N$ & $128$ & $256$ & $512$ \\ 
$\varepsilon$ & $0.14 \times 10^{-6}$ & $0.15 \times 10^{-13}$ & $0.19 \times 10^{-27}$ \\ \hline
\cite{SD02} & $6985$ & $14025$ & $28105$ \\ \hline
Our code & $282$ & $502$ & $934$ \\
($\varepsilon_0 =$) & $(1/2)\varepsilon$ & $(7/10)\varepsilon$ & $(7/10)\varepsilon$ \\ \hline
ratio & $4.04\%$ & $3.58\%$ & $3.32\%$ \\
\end{tabular}
\end{table}
\begin{table}[htb]
\centering
\caption{Comparison of code lengths with the codes by Nuida et al.\ \cite{NFHKWOI09} ($c = 3$)}
\label{tab:comparison_Nui09}
\begin{tabular}{c||c|c|c}
$N$ & $300$ & $10^9$ & $10^6$ \\
$\varepsilon$ & $10^{-11}$ & $10^{-6}$ & $10^{-3}$ \\ \hline
\cite{NFHKWOI09} & $1309$ & $1423$ & $877$ \\ \hline
Our code & $420$ & $556$ & $349$ \\
($\varepsilon_0 =$) & $(9/10)\varepsilon$ & $(1/100)\varepsilon$ & $(1/100)\varepsilon$ \\ \hline
ratio & $32.1\%$ & $39.1\%$ & $39.8\%$ \\
\end{tabular}
\end{table}

\section{Security Proof}
\label{sec:proof}

In this section, we present an outline of the proof of Theorems \ref{thm:error_probability} and \ref{thm:error_probability_p_1/2}.
Omitted details of the proof will be supplied in Sect.~\ref{sec:proof_detail}.

First, we present some properties of the threshold parameter $Z = Z_{y'}$, which will be proven in Sect.~\ref{subsec:proof_threshold}:
\begin{proposition}
\label{prop:threshold}
\begin{enumerate}
\item If $Z$ satisfies the condition (\ref{eq:threshold}), then the conditional probability that $S(\mathsf{I}) \geq Z$ for some $\mathsf{I} \in U_{\mathrm{I}}$, conditioned on the choice of $y'$, is not higher than $(N-1)\varepsilon_0 / N$.
\item The value $Z = Z_0$ in (\ref{eq:threshold_example}) satisfies the condition (\ref{eq:threshold}).
\end{enumerate}
\end{proposition}

To prove Theorem \ref{thm:error_probability}, we consider the case that the number of pirates $|U_{\mathrm{P}}|$ is three.
By symmetry, we may assume that $U_{\mathrm{P}} = \{1,2,3\}$.
Put $T_{\mathrm{P}} = 123$, therefore we have $T_{\mathrm{P}} \in \mathcal{T}(y')$ by Marking Assumption.
Now we consider the following four kinds of events:
\begin{description}
\item[Type I error:] $S(\mathsf{I}) \geq Z$ for some innocent user $\mathsf{I} \in U_{\mathrm{I}}$.
\item[Type II error:] $T \cap T_{\mathrm{P}} = \emptyset$ for some $T \in \mathcal{T}(y')$.
\item[Type III error:] There are $T_1,T_2 \in \mathcal{T}(y')$ such that $\emptyset \neq T_1 \cap T_2 \subseteq U_{\mathrm{I}}$, $|T_1 \cap T_{\mathrm{P}}| = 1$ and $|T_2 \cap T_{\mathrm{P}}| = 1$.
\item[Type IV error:] $S(i) < Z$ for every $i \in \{1,2,3\}$, and there is an innocent user $\mathsf{I}$ such that $12\mathsf{I} \in \mathcal{T}(y')$, $13\mathsf{I} \in \mathcal{T}(y')$ and $23\mathsf{I} \in \mathcal{T}(y')$.
\end{description}
Then we have the following property, which will be proven in Sect.~\ref{subsec:proof_error_type}:
\begin{proposition}
\label{prop:error_type}
If $|U_{\mathrm{P}}| = 3$, then tracing error occurs only when one of the Type I, II, III and IV errors occurs.
\end{proposition}
By this proposition, the error probability is bounded by the sum of the probabilities of Type I--IV errors.
By Proposition \ref{prop:threshold}, the probability of Type I error is bounded by $\varepsilon_0$.
Now Theorem \ref{thm:error_probability} is proven by combining this with the following three propositions, which will be proven in Sect.~\ref{subsec:proof_probability_type_II}, Sect.~\ref{subsec:proof_probability_type_III} and Sect.~\ref{subsec:proof_probability_type_IV}, respectively (see (\ref{eq:thm_error_probability_statement_notation}) for the notations):
\begin{proposition}
\label{prop:probability_type_II}
If $|U_{\mathrm{P}}| = 3$, then the probability of Type II error is not higher than $\binom{N-3}{3} f_1(p)^m$.
\end{proposition}
\begin{proposition}
\label{prop:probability_type_III}
If $|U_{\mathrm{P}}| = 3$, then the probability of Type III error is not higher than $3(N-3)(N-4) f_2(p)^m$.
\end{proposition}
\begin{proposition}
\label{prop:probability_type_IV}
If $|U_{\mathrm{P}}| = 3$ and the threshold $Z$ is chosen so that the condition (\ref{eq:threshold}) holds and $Z \leq Z_0$, then the probability of Type IV error is lower than $(N-3) (1-p)^{-3 \sqrt{ (m/2) \log(N/\varepsilon_0) } } f_3(p)^m$.
\end{proposition}

To prove Theorem \ref{thm:error_probability_p_1/2}, we set $p = 1/2$.
Then the bound of error probability given by Theorem \ref{thm:error_probability} is specialized to the value specified in Theorem \ref{thm:error_probability_p_1/2}.
Hence our remaining task is to evaluate the error probabilities for the case that the number of pirates is one or two.

First we consider the case that there are exactly two pirates, say, $1,2 \in U$.
The key property is the following, which will be proven in Sect.~\ref{subsec:proof_two_pirates_score}:
\begin{proposition}
\label{prop:two_pirates_score}
In this situation, if the condition (\ref{eq:error_probability_condition_m}) is satisfied, then the probability that $S(1) < Z$ and $S(2) < Z$ is lower than $\varepsilon_0 / N$.
\end{proposition}
By this proposition, when the condition (\ref{eq:error_probability_condition_m}) is satisfied, at least one of the two pirates is output in Step \ref{item:tracing_score_check} of the tracing algorithm with probability not lower than $1 - \varepsilon_0 / N$.
On the other hand, by Proposition \ref{prop:threshold}, some innocent user is output in Step \ref{item:tracing_score_check} with probability not higher than $(N-1) \varepsilon_0 / N$.
Hence in Step \ref{item:tracing_score_check}, at least one pirate and no innocent users are output with probability not lower than $1 - \varepsilon_0$.
This implies that the error probability is bounded by $\varepsilon_0$ in this case.

Secondly, we consider the case that there is exactly one pirate, say, $1 \in U$.
Then we have the following property, which will be proven in Sect.~\ref{subsec:proof_one_pirate_score}:
\begin{proposition}
\label{prop:one_pirate_score}
In this situation, if $m \geq 2 \log(N/\varepsilon_0)$, then the score $S(1)$ of the pirate is always higher than or equal to $Z$.
\end{proposition}
By this proposition, when the condition (\ref{eq:error_probability_condition_m}) is satisfied, the pirate is always output in Step \ref{item:tracing_score_check} of the tracing algorithm.
Hence by the same argument as the previous paragraph, the error probability is bounded by $\varepsilon_0$ in this case as well.
Summarizing, the proof of Theorem \ref{thm:error_probability_p_1/2} is concluded.

\section{On implementation of the tracing algorithm}
\label{sec:implementation_tracing_algorithm}

In this section, we discuss some implementation issue of the tracing algorithm $\mathsf{Tr}$ of the proposed $3$-secure code.
More precisely, we consider the calculation of the set $\mathcal{T}(y')$ appeared in Step \ref{item:tracing_triple_compute} of $\mathsf{Tr}$.
By a naive calculation method based on the definition (\ref{eq:definition_parent_triple}) of $\mathcal{T}(y')$, we need to check the condition $y' \in \mathcal{E}(w_{i_1},w_{i_2},w_{i_3})$ for every triple $i_1i_2i_3$ of users, therefore the time complexity with respect to the user number $N$ is inevitably $\Omega(N^3)$.
As this complexity is larger than tracing algorithms of many other $c$-secure codes such as Tardos codes \cite{Tar08}, it is important to reduce the complexity of calculation of $\mathcal{T}(y')$.

To calculate the collection $\mathcal{T}(y')$, we consider the following algorithm, with codewords $w_1,\dots,w_N$ and the $m$-bit word $y'$ as input:
\begin{enumerate}
\item Set $\mathcal{L}^{(1)}_1 = \{i \in [N] \mid w_{i,1} = y'_1\}$ and $\mathcal{L}^{(1)}_2 = \mathcal{L}^{(1)}_3 = \emptyset$.
\item For each $2 \leq j \leq m$, construct $\mathcal{L}^{(j)}_1$, $\mathcal{L}^{(j)}_2$ and $\mathcal{L}^{(j)}_3$ inductively, in the following manner.
(At the beginning, set $\mathcal{L}^{(j)}_1 = \mathcal{L}^{(j)}_2 = \mathcal{L}^{(j)}_3 = \emptyset$.)
\begin{enumerate}
\item Put $C_j = \{i \in [N] \mid w_{i,j} = y'_j\}$.
\item Set $\mathcal{L}^{(j)}_1 = \mathcal{L}^{(j-1)}_1 \cap C_j$.
\item Add the pair $\{\mathcal{L}^{(j-1)}_1 \setminus C_j,C_j \setminus \mathcal{L}^{(j-1)}_1\}$ of subsets of $[N]$ to $\mathcal{L}^{(j)}_2$.
\item For each pair $\{K_1,K_2\}$ of subsets of $[N]$ in $\mathcal{L}^{(j-1)}_2$,
\begin{itemize}
\item add two pairs $\{K_1 \cap C_j,K_2\}$ and $\{K_1 \setminus C_j,\allowbreak K_2 \cap C_j\}$ to $\mathcal{L}^{(j)}_2$;
\item add the triple $\{K_1 \setminus C_j,K_2 \setminus C_j,C_j \setminus (K_1 \cup K_2)\}$ of subsets of $[N]$ to $\mathcal{L}^{(j)}_3$.
\end{itemize}
\item For each triple $\{K_1,K_2,K_3\}$ of subsets of $[N]$ in $\mathcal{L}^{(j-1)}_3$, add three triples $\{K_1 \cap C_j,K_2,K_3\}$, $\{K_1 \setminus C_j,K_2 \cap C_j,K_3\}$, $\{K_1 \setminus C_j,K_2 \setminus C_j,K_3 \cap C_j\}$ to $\mathcal{L}^{(j)}_3$.
\item Remove from $\mathcal{L}^{(j)}_2$ every pair $\{K_1,K_2\}$ with $K_1$ or $K_2$ being empty, and from $\mathcal{L}^{(j)}_3$ every triple $\{K_1,K_2,K_3\}$ with $K_1$, $K_2$ or $K_3$ being empty.
\end{enumerate}
\item Output the collection of the triples $T = i_1 i_2 i_3$ of distinct numbers $i_1,i_2,i_3$ satisfying one of the following conditions:
\begin{itemize}
\item we have $i_1 \in \mathcal{L}^{(m)}_1$ and $i_2,i_3$ are arbitrary;
\item for some $\{K_1,K_2\} \in \mathcal{L}^{(m)}_2$, we have $i_1 \in K_1$, $i_2 \in K_2$ and $i_3$ is arbitrary;
\item for some $\{K_1,K_2,K_3\} \in \mathcal{L}^{(m)}_3$, we have $i_1 \in K_1$, $i_2 \in K_2$ and $i_3 \in K_3$.
\end{itemize}
\end{enumerate}
An inductive argument shows that, for each $j \in [m]$ and each triple of distinct $i_1,i_2,i_3$, the $j$-bit initial subword $(y'_1,\dots,y'_j)$ of $y'$ is in the envelope of the $j$-bit initial subwords of $w_{i_1},w_{i_2},w_{i_3}$ if and only if one of the following conditions is satisfied (note that the order of members of a pair or triple is ignored):
\begin{itemize}
\item we have $i_1 \in \mathcal{L}^{(j)}_1$ and $i_2,i_3$ are arbitrary;
\item for some $\{K_1,K_2\} \in \mathcal{L}^{(j)}_2$, we have $i_1 \in K_1$, $i_2 \in K_2$ and $i_3$ is arbitrary;
\item for some $\{K_1,K_2,K_3\} \in \mathcal{L}^{(j)}_3$, we have $i_1 \in K_1$, $i_2 \in K_2$ and $i_3 \in K_3$.
\end{itemize}
By setting $j = m$, it follows that the above algorithm outputs $\mathcal{T}(y')$ correctly.

Now for each $2 \leq j \leq m$, complexity of computing $\mathcal{L}^{(j)}_1$, $\mathcal{L}^{(j)}_2$, and $\mathcal{L}^{(j)}_3$ from $\mathcal{L}^{(j-1)}_1$, $\mathcal{L}^{(j-1)}_2$, and $\mathcal{L}^{(j-1)}_3$ is approximately proportional to $N$ times the total number of members of $\mathcal{L}^{(j-1)}_2$ and $\mathcal{L}^{(j-1)}_3$.
Hence the total complexity of the algorithm is approximately proportional to $N m$ times the average of total number of members in $\mathcal{L}^{(j)}_2$ and $\mathcal{L}^{(j)}_3$ over all $1 \leq j \leq m-1$.
This implies that the order (with respect to $N$) of complexity of calculating $\mathcal{T}(y')$ can be reduced from $\Theta(N^3)$ if the average number of pairs and triples in $\mathcal{L}^{(j)}_2$ and $\mathcal{L}^{(j)}_3$ is sufficiently small.
The author guesses that the latter average number is indeed sufficiently small in most of the practical cases, as the size of $\mathcal{T}(y')$ would be not large in average case (provided the code length $m$ is long enough to make the error probability of the fingerprint code sufficiently small).
A detailed analysis of this calculation method will be a future research topic.
Instead, here we show some experimental data for running time of the above algorithm, which was implemented on a usual PC with $1.83$GHz Intel Core 2 CPU and $2$Gbytes memory.
We chose parameters $N = 1000$, $m = 180$, $\varepsilon_0 = 0.001$, and adopted minority vote attack as pirate strategy.
Then the average running time of the algorithm over $10$ trials was $4331.5$ seconds, i.e., about $1$ hour and $13$ minutes, where the calculation of running times was restricted to the case that scores of all users are less than the threshold, as otherwise the tracing algorithm halts before Step \ref{item:tracing_triple_compute}.

\section{Proofs of the Propositions}
\label{sec:proof_detail}

\subsection{Proof of Proposition \ref{prop:threshold}}
\label{subsec:proof_threshold}

First, we prove the claim 1 of Proposition \ref{prop:threshold}.
For each $\mathsf{I} \in U_{\mathrm{I}}$ and $\sigma \in \{\mathrm{H},\mathrm{L}\}$, let $K_{\sigma} = \{j \in A_{\sigma} \mid w_{\mathsf{I},j} = y'_j\}$.
Then we have $S(\mathsf{I}) = |K_{\mathrm{H}}| \log(1/p) + |K_{\mathrm{L}}| \log(1/(1-p))$.
Now note that the choice of $y'$ is independent of $w_{\mathsf{I}}$.
This implies that we have $Pr[ w_{\mathsf{I},j} = y'_j \mid y' ] = p$ for each $j \in A_{\mathrm{H}}$, and we have $Pr[ w_{\mathsf{I},j} = y'_j \mid y' ] = 1-p$ for each $j \in A_{\mathrm{L}}$.
Hence the conditional probability that $|K_{\mathrm{H}}| = k_{\mathrm{H}}$ and $|K_{\mathrm{L}}| = k_{\mathrm{L}}$, conditioned on this $y'$, is $\binom{a_{\mathrm{L}}}{k_{\mathrm{L}}} (1-p)^{k_{\mathrm{L}}} p^{a_{\mathrm{L}} - k_{\mathrm{L}}} \binom{a_{\mathrm{H}}}{k_{\mathrm{H}}} p^{k_{\mathrm{H}}} (1-p)^{a_{\mathrm{H}} - k_{\mathrm{H}}}$.
This implies that $Pr[ S(\mathsf{I}) \geq Z \mid y' ]$ is equal to the left-hand side of (\ref{eq:threshold}), therefore the claim 1 holds as there exist at most $N-1$ innocent users $\mathsf{I}$.

Secondly, to prove the claim 2 of Proposition \ref{prop:threshold}, we use the following Hoeffding's Inequality:
\begin{theorem}
[{\cite{Hoe63}, Theorem 2}]
\label{thm:Hoeffding}
Let $X_1,X_2,\dots,X_n$ be independent random variables such that $a_i \leq X_i \leq b_i$ for each $i$.
Let $\overline{X}$ be the average value of $X_1,\dots,X_n$.
Then for $t > 0$, we have
\begin{equation}
Pr[ \overline{X} - E[\overline{X}] \geq t ] \leq \exp\left( \frac{ -2 n^2 t^2 }{ \sum_{i=1}^{n} (b_i - a_i)^2 } \right) \enspace.
\end{equation}
\end{theorem}

As mentioned above, the left-hand side of (\ref{eq:threshold}) is equal to $Pr[ S(\mathsf{I}) \geq Z \mid y' ]$, where $\mathsf{I}$ is any specified innocent user.
Now for each $j \in [m]$, let $X_j$ be a random variable such that
\begin{equation}
\begin{cases}
Pr[ X_j = \log(1/p) ] = p \,,\, Pr[ X_j = 0 ] = 1 - p &\mbox{if } j \in A_{\mathrm{H}} \enspace,\\
Pr[ X_j = \log(1/(1-p)) ] = 1 - p \,,\, Pr[ X_j = 0 ] = p &\mbox{if } j \in A_{\mathrm{L}} \enspace.
\end{cases}
\end{equation}
Then, conditioned on this $y'$, the variables $X_1,\dots,X_m$ are independent and $S(\mathsf{I}) = m \overline{X}$.
Now by a direct calculation, we have $E[S(\mathsf{I}) \mid y'] = m E[\overline{X} \mid y'] = \mu$ where $\mu = a_{\mathrm{H}} p \log(1/p) + a_{\mathrm{L}} (1-p) \log(1/(1-p))$.
Moreover, we have $0 \leq X_j \leq \log(1/p)$ if $j \in A_{\mathrm{H}}$, and we have $0 \leq X_j \leq \log(1/(1-p))$ if $j \in A_{\mathrm{L}}$.
Hence Theorem \ref{thm:Hoeffding} implies that
\begin{equation}
\label{eq:proof_threshold_1}
Pr[ S(\mathsf{I}) - \mu \geq m t \mid y' ]
\leq \exp\left( \frac{ -2 m^2 t^2 }{ a_{\mathrm{H}}(\log(1/p))^2 + a_{\mathrm{L}}(\log(1/(1-p)))^2 } \right)
\end{equation}
for $t > 0$.
Now by setting $t = \eta / m$ where
\begin{equation}
\eta = \sqrt{ \frac{1}{2} \left( \left( \log\frac{1}{p} \right)^2 a_{\mathrm{H}} + \left( \log\frac{1}{1-p} \right)^2 a_{\mathrm{L}} \right) \log \frac{N}{\varepsilon_0} } \enspace,
\end{equation}
the right-hand side of (\ref{eq:proof_threshold_1}) is equal to $\varepsilon_0 / N$.
On the other hand, for the left-hand side of (\ref{eq:proof_threshold_1}), we have
\begin{equation}
Pr[ S(\mathsf{I}) - \mu \geq m t \mid y' ] = Pr[ S(\mathsf{I}) \geq \mu + \eta \mid y' ] \enspace,
\end{equation}
while the value of $Z = Z_0$ in (\ref{eq:threshold_example}) is equal to $\mu + \eta$.
Hence the condition (\ref{eq:threshold}) is satisfied, concluding the proof of Proposition \ref{prop:threshold}.

\subsection{Proof of Proposition \ref{prop:error_type}}
\label{subsec:proof_error_type}

To prove Proposition \ref{prop:error_type}, suppose that it is not the case of Type I--IV errors.
We show that tracing error does not occur in this case.
Recall that $T_{\mathrm{P}} = 123 \in \mathcal{T}(y')$.
By the absence of Type I error, it holds that either some pirate and no innocent users are output in Step \ref{item:tracing_score_check} of $\mathsf{Tr}$, or $S(i) < Z$ for every $i \in U$ and nobody is output in Step \ref{item:tracing_score_check}.
It suffices to consider the latter case.
We have $T_{\mathrm{P}} \in \mathcal{T}'$ by the absence of Type II error.
Hence every $T \in \mathcal{T}'$ intersects $T_{\mathrm{P}}$, and $\bigcap \mathcal{T}' \subseteq T_{\mathrm{P}}$.
By virtue of Step \ref{item:tracing_triple_intersection}, it suffices to consider the case that $\bigcap \mathcal{T}' = \emptyset$.
Now there are the following two cases: (A) we have $|T \cap T_{\mathrm{P}}| = 1$ for some $T \in \mathcal{T}'$; (B) we have $|T \cap T_{\mathrm{P}}| = 2$ for every $T \in \mathcal{T}' \setminus \{T_{\mathrm{P}}\}$.

\subsubsection{Case (A)}
\label{subsec:error_type_A}

Let $T_1 \in \mathcal{T}'$ and $|T_1 \cap T_{\mathrm{P}}| = 1$.
By symmetry, we may assume that $T_1 \cap T_{\mathrm{P}} = \{1\}$.
By the fact $\bigcap \mathcal{T}' = \emptyset$, there is a $T_2 \in \mathcal{T}'$ such that $1 \not\in T_2$.
We may assume by symmetry that $2 \in T_2$, as $T_2 \cap T_{\mathrm{P}} \neq \emptyset$.
We have $T_1 \cap T_2 \neq \emptyset$ as $T_1 \in \mathcal{T}'$, therefore the absence of Type III error implies that $3 \in T_2$.
Put $T_2 = 23\mathsf{I}$ with $\mathsf{I} \in U_{\mathrm{I}}$, and $T_1 = 1 \mathsf{I} \mathsf{I}'$ with $\mathsf{I}' \in U_{\mathrm{I}}$.
Now if we calculate the set $\mathcal{P}$ by using $\{T_{\mathrm{P}},T_1,T_2\}$ instead of $\mathcal{T}'$, then the result is
\begin{equation}
\label{eq:proof_error_type_A_setP}
\{12, 13, 1\mathsf{I}, 2\mathsf{I}, 2\mathsf{I}', 3\mathsf{I}, 3\mathsf{I}'\} \enspace.
\end{equation}
In general, the actual set $\mathcal{P}$ is included in the set (\ref{eq:proof_error_type_A_setP}).
Now we present two properties.
First, we show that $12,13 \in \mathcal{P}$.
Indeed, if $12 \not\in \mathcal{P}$, then we have $12 \cap T = \emptyset$ for some $t \in \mathcal{T}'$.
Now we have $3 \in T$ and $T_1 \cap T \neq \emptyset$ as $T \in \mathcal{T}'$, therefore $T_1$ and $T$ contradict the absence of Type III error.
Hence we have $12 \in \mathcal{P}$, and $13 \in \mathcal{P}$ by symmetry.
Secondly, we show that no innocent users are output in Step \ref{item:tracing_pair_multiplicity_1}.
Indeed, if an $\mathsf{I}'' \in U_{\mathrm{I}}$ is output in Step \ref{item:tracing_pair_multiplicity_1}, then the possibility of $\mathcal{P}$ mentioned above implies that $\mathsf{I}'' \in \{\mathsf{I},\mathsf{I}'\}$ and we have $i \in \mathcal{P}_1$ and $i \mathsf{I}'' \in \mathcal{P}$ for some $i \in 123$.
This is impossible, as $12,13 \in \mathcal{P}$.
Hence this claim holds, therefore it suffices to consider the case that nobody is output in Step \ref{item:tracing_pair_multiplicity_1}, namely $\mathcal{P}_1 = \emptyset$.

By these properties, we have either $2 \mathsf{I}',3 \mathsf{I}' \in \mathcal{P}$ or  $2 \mathsf{I}',3 \mathsf{I}' \not\in \mathcal{P}$ (otherwise $\mathsf{I}' \in \mathcal{P}_1$, a contradiction).
Similarly, we have $\mathcal{P} \cap \{2 \mathsf{I},2 \mathsf{I}'\} \neq \emptyset$ and $\mathcal{P} \cap \{3 \mathsf{I},3 \mathsf{I}'\} \neq \emptyset$.
First we consider the case that $2 \mathsf{I}', 3\mathsf{I}' \in \mathcal{P}$.
As $\mathcal{P}_1 = \emptyset$, it does not hold that $|\mathcal{P} \cap \{1 \mathsf{I},2 \mathsf{I},3 \mathsf{I}\}| \neq 1$.
If $1 \mathsf{I},2 \mathsf{I},3 \mathsf{I} \in \mathcal{P}$, then $|\mathcal{P}| = 7$, $\mathcal{P}_2 = \{\mathsf{I}'\}$, and $2$ and $3$ are output in Step \ref{item:tracing_pair_7}.
If $|\mathcal{P} \cap \{1 \mathsf{I},2 \mathsf{I},3 \mathsf{I}\}| = 2$, then $|\mathcal{P}| = 6$, $\emptyset \neq \mathcal{P}_3 \subseteq U_{\mathrm{P}}$ and a pirate is correctly output in Step \ref{item:tracing_pair_6}.
Finally, if $1 \mathsf{I},2 \mathsf{I},3 \mathsf{I} \not\in \mathcal{P}$, then $|\mathcal{P}| = 4$ and $\mathcal{P} = \{12,13,2\mathsf{I}',3\mathsf{I}'\}$.
Now $\mathsf{I}'$ is not output in Step \ref{item:tracing_pair_4}, as $123 \in \mathcal{T}'$.
Moreover, if none of $1$, $2$, and $3$ is output in Step \ref{item:tracing_pair_4}, then it should hold that $12\mathsf{I}', 13\mathsf{I}', 23\mathsf{I}' \in \mathcal{T}'$, contradicting the absence of Type IV error.
Hence a pirate is correctly output in Step \ref{item:tracing_pair_4}, concluding the proof in the case $2\mathsf{I}',3\mathsf{I}' \in \mathcal{P}$.

Secondly, we suppose that $2\mathsf{I}',3\mathsf{I}' \not\in \mathcal{P}$, therefore $2\mathsf{I},3\mathsf{I} \in \mathcal{P}$.
There are two possibilities $\mathcal{P} = \{12,13,2\mathsf{I},3\mathsf{I}\}$ and $\mathcal{P} = \{12,13,1\mathsf{I},2\mathsf{I},3\mathsf{I}\}$.
The former case is the same as the previous paragraph.
In the latter case, we have $|\mathcal{P}| = 5$, $\mathcal{T}'' \subseteq \{12\mathsf{I},13\mathsf{I}\}$ and $\mathcal{P}_2 = 23$.
Hence $2$ or $3$ is correctly output in Step \ref{item:tracing_pair_5_nonempty} when $\mathcal{T}'' \neq \emptyset$.
On the other hand, when $\mathcal{T}'' = \emptyset$, $2$ and $3$ are correctly output in Step \ref{item:tracing_pair_5_empty}.
Hence the proof in the case $2\mathsf{I}',3\mathsf{I}' \not\in \mathcal{P}$ (therefore in the case (A)) is concluded.

\subsubsection{Case (B)}
\label{subsec:error_type_B}

As $\bigcap \mathcal{T}' = \emptyset$, there are $\mathsf{I}_1,\mathsf{I}_2,\mathsf{I}_3 \in U_{\mathrm{I}}$ such that $12\mathsf{I}_3,13\mathsf{I}_2,\allowbreak 23\mathsf{I}_1 \in \mathcal{T}'$.
By the absence of Type IV error, it does not hold that $\mathsf{I}_1 = \mathsf{I}_2 = \mathsf{I}_3$.
By symmetry, we may assume that $\mathsf{I}_1 \neq \mathsf{I}_2$.
Then by calculating the set $\mathcal{P}$ by using $\{123,12\mathsf{I}_3,13\mathsf{I}_2,23\mathsf{I}_1\}$ instead of $\mathcal{T}'$, it follows that the actual $\mathcal{P}$ satisfies $\mathcal{P} \subseteq \{12,13,23,1\mathsf{I}_1,2\mathsf{I}_2,3\mathsf{I}_3\}$, while $12,13,23 \in \mathcal{P}$ by the assumption of the case (B).
If $\mathcal{P} = \{12,13,23\}$, then $1$, $2$ and $3$ are output in Step \ref{item:tracing_pair_3}.
Therefore it suffices to consider the case that $\{12,13,23\} \subsetneq \mathcal{P}$.

If $\mathsf{I}_1 \neq \mathsf{I}_3 \neq \mathsf{I}_2$, then we have $\emptyset \neq \mathcal{P}_1 \subseteq \mathsf{I}_1\mathsf{I}_2\mathsf{I}_3$ and a pirate is correctly output in Step \ref{item:tracing_pair_multiplicity_1}.
Hence it suffices to consider the remaining case.
By symmetry, we may assume that $\mathsf{I}_1 = \mathsf{I}_3 \neq \mathsf{I}_2$.
If $2 \mathsf{I}_2 \in \mathcal{P}$, then we have $\mathsf{I}_2 \in \mathcal{P}_1 \subseteq \mathsf{I}_1\mathsf{I}_2$, and $2$ is correctly output in Step \ref{item:tracing_pair_multiplicity_1}.
From now, we assume that $2\mathsf{I}_2 \not\in \mathcal{P}$.
If $1\mathsf{I}_1 \not\in \mathcal{P}$ or $3\mathsf{I}_1 \not\in \mathcal{P}$, then we have $\mathcal{P}_1 = \{\mathsf{I}_1\}$ as $\{12,13,23\} \subsetneq \mathcal{P}$, therefore $1$ or $3$ is correctly output in Step \ref{item:tracing_pair_multiplicity_1}.
On the other hand, if $1\mathsf{I}_1,3\mathsf{I}_1 \in \mathcal{P}$, then we have $\mathcal{P} = \{12,13,23,1\mathsf{I}_1,3\mathsf{I}_1\}$, while $13\mathsf{I}_1 \not\in \mathcal{T}'$ by the absence of Type IV error (note that $12\mathsf{I}_1,23\mathsf{I}_1 \in \mathcal{T}'$), therefore $\mathcal{T}'' = \{123\}$, $\mathcal{P}_2 = 2 \mathsf{I}_1$ and $2$ is correctly output in Step \ref{item:tracing_pair_5_nonempty}.
Hence the proof in the case (B), therefore the proof of Proposition \ref{prop:error_type}, is concluded.

\subsection{Proof of Proposition \ref{prop:probability_type_II}}
\label{subsec:proof_probability_type_II}

To prove Proposition \ref{prop:probability_type_II}, let $\mathsf{I}_1$, $\mathsf{I}_2$ and $\mathsf{I}_3$ be three distinct innocent users.
Given $y'$ and $\mathsf{st} = (p_j)_j$, we introduce the following notation for $j \in [m]$:
\begin{equation}
\label{eq:definition_two_bits}
\xi_j^{\mathrm{H}} = \begin{cases}
1 & \mbox{if } p_j = p \enspace, \\
0 & \mbox{if } p_j = 1-p \enspace,
\end{cases}
\xi_j^{\mathrm{L}} = 1 - \xi_j^{\mathrm{H}} \enspace.
\end{equation}
Note that the sets $A_{\sigma}$ for $\sigma \in \{\mathrm{H},\mathrm{L}\}$ defined in Sect.~\ref{sec:results} satisfy that $A_{\sigma} = \{j \mid y'_j = \xi_j^{\sigma}\}$.
We write $A_{\sigma} = A_{\sigma}(y',\mathsf{st})$ and $a_{\sigma} = |A_{\sigma}| = a_{\sigma}(y',\mathsf{st})$ when we emphasize the dependency on $y'$ and $\mathsf{st}$.
Then, as the bits of codewords are independently chosen, we have
\begin{equation}
Pr[\mathsf{I}_1\mathsf{I}_2\mathsf{I}_3 \in \mathcal{T}(y') \mid y',\mathsf{st}] = (1-p^3)^{a_{\mathrm{L}}}(1-(1-p)^3)^{a_{\mathrm{H}}} \enspace,
\end{equation}
therefore
\begin{equation}
\label{eq:proof_type_II_probability_triple}
Pr[\mathsf{I}_1\mathsf{I}_2\mathsf{I}_3 \in \mathcal{T}(y')]
= \sum_{y',\mathsf{st}} Pr[y',\mathsf{st}] (1-p^3)^{a_{\mathrm{L}}(y',\mathsf{st})}(1-(1-p)^3)^{a_{\mathrm{H}}(y',\mathsf{st})} \enspace.
\end{equation}
Now we present the following key lemma, which will be proven later:
\begin{lemma}
\label{lem:proof_type_II_majority}
Among the possible pirate strategies $\rho$, the maximum value of the right-hand side of (\ref{eq:proof_type_II_probability_triple}) is attained by the majority vote attack, namely the attack word $y$ for codewords $w_1,w_2,w_3$ of three pirates satisfies that $y_j = 0$ if at least two of $w_{1,j},w_{2,j},w_{3,j}$ are $0$ and $y_j = 1$ otherwise.
\end{lemma}
If $\rho$ is the majority vote attack, then for each $j \in [m]$, we have $j \in A_{\mathrm{H}}(y',\mathsf{st})$ (i.e., $\xi_j^{\mathrm{H}}$ becomes the majority in $w_{1,j},w_{2,j},w_{3,j}$) with probability $3p^2(1-p) + p^3 = 3p^2 - 2p^3$ and $j \in A_{\mathrm{L}}(y',\mathsf{st})$ with probability $1 - 3p^2 + 2p^3$.
This implies that
\begin{equation}
\begin{split}
&Pr[\mathsf{I}_1\mathsf{I}_2\mathsf{I}_3 \in \mathcal{T}(y')] \\
&= \sum_{\substack{\alpha_{\mathrm{L}},\alpha_{\mathrm{H}} \\ \alpha_{\mathrm{L}} + \alpha_{\mathrm{H}} = m}} Pr[a_{\mathrm{L}} = \alpha_{\mathrm{L}}, a_{\mathrm{H}} = \alpha_{\mathrm{H}}] (1-p^3)^{\alpha_{\mathrm{L}}}(1-(1-p)^3)^{\alpha_{\mathrm{H}}} \\
&= \sum_{\substack{\alpha_{\mathrm{L}},\alpha_{\mathrm{H}} \\ \alpha_{\mathrm{L}} + \alpha_{\mathrm{H}} = m}} \Biggl( \binom{m}{\alpha_{\mathrm{L}}} (1-3p^2+2p^3)^{\alpha_{\mathrm{L}}} (3p^2-2p^3)^{\alpha_{\mathrm{H}}} (1-p^3)^{\alpha_{\mathrm{L}}}(1-(1-p)^3)^{\alpha_{\mathrm{H}}} \Biggr) \\
&= \sum_{\substack{\alpha_{\mathrm{L}},\alpha_{\mathrm{H}} \\ \alpha_{\mathrm{L}} + \alpha_{\mathrm{H}} = m}} \Biggl( \binom{m}{\alpha_{\mathrm{L}}} (1-3p^2+p^3+3p^5-2p^6)^{\alpha_{\mathrm{L}}} (9p^3-15p^4+9p^5-2p^6)^{\alpha_{\mathrm{H}}} \Biggr) \\
&= (1 - 3p^2 + 10p^3 - 15p^4 + 12p^5 - 4p^6)^m = f_1(p)^m \enspace.
\end{split}
\end{equation}
By virtue of Lemma \ref{lem:proof_type_II_majority}, for a general $\rho$, $Pr[\mathsf{I}_1\mathsf{I}_2\mathsf{I}_3 \in \mathcal{T}(y')]$ is bounded by the right-hand side of the above equality.
This implies the claim of Proposition \ref{prop:probability_type_II}, as there are $\binom{N-3}{3}$ choices of the triple $\mathsf{I}_1,\mathsf{I}_2,\mathsf{I}_3$.

To complete the proof of Proposition \ref{prop:probability_type_II}, we give a proof of Lemma \ref{lem:proof_type_II_majority}.
\begin{proof}
[Proof of Lemma \ref{lem:proof_type_II_majority}]
Fix the codewords $w_1,w_2,w_3$ of the three pirates $1,2,3 \in U$.
Let $\vec{w}_{\mathrm{P}}$ denote the collection of those three codewords.
Let $j_0 \in [m]$ be the index of a detectable column.
By symmetry, we may assume without loss of generality that $w_{1,j_0} = w_{2,j_0} = 0$ and $w_{3,j_0} = 1$.
Now let $y^0$ be an arbitrary attack word such that $y^0_{j_0} = 0$, and let $y^1$ and $y^{?}$ be the attack words obtained from $y^0$ by changing the $j_0$-th column to $1$ and to ${?}$, respectively.
We show that if the pirate strategy $\rho$ for the input $\vec{w}_{\mathrm{p}}$ is modified so that it outputs $y^0$ instead of $y^1$ and $y^{?}$, then the right-hand side of (\ref{eq:proof_type_II_probability_triple}) will not decrease.
As $\vec{w}_{\mathrm{P}}$, $j_0$ and $y^0$ are arbitrarily chosen, the claim of Lemma \ref{lem:proof_type_II_majority} then follows.

Let $y'{}^0$ be an $m$-bit word such that $y'{}^0_j = y^0_j$ for any $j \in [m]$ with $y^0_j \neq {?}$, therefore $y'{}^0$ is obtained from $y^0$ in Step \ref{item:tracing_replace_?} in the tracing algorithm with positive probability.
Let $y'{}^1$ be the $m$-bit word obtained from $y'{}^0$ by changing the $j_0$-th column to $1$.
Moreover, let $\mathsf{st}^0 = (p_j)_j$ be any state information such that $p_{j_0} = 1 - p$, and let $\mathsf{st}^1$ be the state information obtained from $\mathsf{st}^0$ by changing the $j_0$-th component to $p$.

In this case, by independence of the columns, we have $Pr[\vec{w}_{\mathrm{P}} \mid \mathsf{st}^0] = \alpha p^2(1-p)$ and $Pr[\vec{w}_{\mathrm{P}} \mid \mathsf{st}^1] = \alpha p(1-p)^2$ for a common $\alpha > 0$.
As $Pr[\mathsf{st}^0] = Pr[\mathsf{st}^1] > 0$ and $Pr[\vec{w}_{\mathrm{P}}] > 0$, Bayes Theorem implies that $Pr[\mathsf{st}^0 \mid \vec{w}_{\mathrm{P}}] = \alpha' p^2(1-p)$ and $Pr[\mathsf{st}^1 \mid \vec{w}_{\mathrm{P}}] = \alpha' p(1-p)^2$ for a common $\alpha' > 0$, therefore
\begin{equation}
Pr[\mathsf{st}^0 \mid \vec{w}_{\mathrm{P}}, \mbox{ ($\mathsf{st}^0$ or $\mathsf{st}^1$) }] = \frac{ \alpha' p^2(1-p) }{ \alpha' p^2(1-p) + \alpha' p(1-p)^2 } = p
\end{equation}
and $Pr[\mathsf{st}^1 \mid \vec{w}_{\mathrm{P}}, \mbox{ ($\mathsf{st}^0$ or $\mathsf{st}^1$) }] = 1-p$.
Now there is a common $\beta > 0$ such that, for each $x \in \{0,1\}$,
\begin{equation}
\begin{split}
Pr[y'{}^0 \mid \mathsf{st}^x, y^0] &= Pr[y'{}^1 \mid \mathsf{st}^x, y^1] = \beta \enspace,\\
Pr[y'{}^0 \mid \mathsf{st}^x, y^1] &= Pr[y'{}^1 \mid \mathsf{st}^x, y^0] = 0 \enspace,\\
Pr[y'{}^0 \mid \mathsf{st}^0, y^{?}] &= Pr[y'{}^1 \mid \mathsf{st}^1, y^{?}] = \beta p \enspace,\\
Pr[y'{}^0 \mid \mathsf{st}^1, y^{?}] &= Pr[y'{}^1 \mid \mathsf{st}^0, y^{?}] = \beta (1-p) \enspace.
\end{split}
\end{equation}
As the choice of the attack word $y$ for given $\vec{w}_{\mathrm{P}}$ is independent of $\mathsf{st}$, and the choice of the word $y'$ will be independent of $\vec{w}_{\mathrm{P}}$ once the attack word $y$ is determined, it follows that
\begin{equation}
Pr[y'{}^x,\mathsf{st}^{x'} \mid \vec{w}_{\mathrm{P}}, \mbox{ ($\mathsf{st}^0$ or $\mathsf{st}^1$) }, y^{x''}] = Pr[\mathsf{st}^{x'} \mid \vec{w}_{\mathrm{P}}, \mbox{ ($\mathsf{st}^0$ or $\mathsf{st}^1$) }] Pr[y'{}^x \mid \mathsf{st}^{x'}, y^{x''}]
\end{equation}
for $x,x' \in \{0,1\}$ and $x'' \in \{0,1,?\}$.
By these relations, we have
\begin{equation}
\begin{split}
&Pr[(y'{}^0,\mathsf{st}^0) \mbox{ or } (y'{}^1,\mathsf{st}^1) \mid \vec{w}_{\mathrm{P}}, \mbox{ ($\mathsf{st}^0$ or $\mathsf{st}^1$) }, y^0]
= p \cdot \beta + (1-p) \cdot 0 = p \beta \enspace, \\
&Pr[(y'{}^1,\mathsf{st}^0) \mbox{ or } (y'{}^0,\mathsf{st}^1) \mid \vec{w}_{\mathrm{P}}, \mbox{ ($\mathsf{st}^0$ or $\mathsf{st}^1$) }, y^0]
= 1 - p \beta \enspace, \\
&Pr[(y'{}^0,\mathsf{st}^0) \mbox{ or } (y'{}^1,\mathsf{st}^1) \mid \vec{w}_{\mathrm{P}}, \mbox{ ($\mathsf{st}^0$ or $\mathsf{st}^1$) }, y^1]
= p \cdot 0 + (1-p) \cdot \beta = (1-p) \beta \enspace, \\
&Pr[(y'{}^1,\mathsf{st}^0) \mbox{ or } (y'{}^0,\mathsf{st}^1) \mid \vec{w}_{\mathrm{P}}, \mbox{ ($\mathsf{st}^0$ or $\mathsf{st}^1$) }, y^1]
= 1 - (1-p) \beta \enspace, \\
&Pr[(y'{}^0,\mathsf{st}^0) \mbox{ or } (y'{}^1,\mathsf{st}^1) \mid \vec{w}_{\mathrm{P}}, \mbox{ ($\mathsf{st}^0$ or $\mathsf{st}^1$) }, y^{?}]
= p \cdot \beta p + (1-p) \cdot \beta p = p \beta \enspace, \\
&Pr[(y'{}^1,\mathsf{st}^0) \mbox{ or } (y'{}^0,\mathsf{st}^1) \mid \vec{w}_{\mathrm{P}}, \mbox{ ($\mathsf{st}^0$ or $\mathsf{st}^1$) }, y^{?}]
= 1 - p \beta \enspace.
\end{split}
\end{equation}
Now note that $p \geq 1/2$, therefore we have $1 - p^3 \leq 1 - (1-p)^3$ and $p \beta \geq (1-p) \beta$.
Note also that $a_{\mathrm{H}}(y'{}^0,\mathsf{st}^0) = a_{\mathrm{H}}(y'{}^1,\mathsf{st}^1) = a_{\mathrm{H}}(y'{}^0,\mathsf{st}^1) + 1 = a_{\mathrm{H}}(y'{}^1,\mathsf{st}^0) + 1$.
This implies that, in the case $\mathsf{st} \in \{\mathsf{st}^0,\mathsf{st}^1\}$, if the pirate strategy $\rho$ for the input $\vec{w}_{\mathrm{p}}$ is modified in such a way that it outputs $y^0$ instead of $y^1$ and $y^{?}$, then the right-hand side of (\ref{eq:proof_type_II_probability_triple}) will not decrease.
As this property is in fact independent of the choice of $\mathsf{st}^0$ and $\mathsf{st}^1$, the claim in the proof follows, concluding the proof of Lemma \ref{lem:proof_type_II_majority}.
\end{proof}

\subsection{Proof of Proposition \ref{prop:probability_type_III}}
\label{subsec:proof_probability_type_III}

To prove Proposition \ref{prop:probability_type_III}, we fix an innocent user $\mathsf{I}_0 \in U_{\mathrm{I}}$ and consider the probability that there are $T_1,T_2 \in \mathcal{T}(y')$ such that $\mathsf{I}_0 \in T_1 \cap T_2 \subseteq U_{\mathrm{I}}$, $T_1 \cap T_{\mathrm{P}} = \{1\}$ and $T_2 \cap T_{\mathrm{P}} = \{2\}$; or equivalently, there are innocent users $\mathsf{I}_1,\mathsf{I}_2 \in U_{\mathrm{I}} \setminus \{\mathsf{I}_0\}$ such that $1 \mathsf{I}_0 \mathsf{I}_1 \in \mathcal{T}(y')$ and $2 \mathsf{I}_0 \mathsf{I}_2 \in \mathcal{T}(y')$.
We introduce some notations.
Given $y'$, $w_1$, $w_2$, $w_{\mathrm{I}_0}$, and $\mathsf{st} = (p_j)_j$, we define, for $\alpha,\beta,\gamma,\delta \in \{\mathrm{H},\mathrm{L}\}$,
\begin{equation}
a_{\alpha \beta \gamma \delta} = |\{j \in [m] \mid y'_j = \xi^{\alpha}_j, w_{1,j} = \xi^{\beta}_j, w_{2,j} = \xi^{\gamma}_j, w_{\mathrm{I}_0,j} = \xi^{\delta}_j\}|
\end{equation}
(see (\ref{eq:definition_two_bits}) for the notations).
Moreover, by using `$*$' as a wild-card, we extend naturally the definition of $a_{\alpha \beta \gamma \delta}$ to the case $\alpha,\beta,\gamma,\delta \in \{\mathrm{H},\mathrm{L},*\}$.
For example, we have $a_{\alpha * * \delta} = a_{\alpha \mathrm{H} \mathrm{H} \delta} + a_{\alpha \mathrm{H} \mathrm{L} \delta} + a_{\alpha \mathrm{L} \mathrm{H} \delta} + a_{\alpha \mathrm{L} \mathrm{L} \delta}$.
Note that $a_{x * * *}$ ($x \in \{\mathrm{H},\mathrm{L}\}$) is equal to the value $a_x$ in Sect.~\ref{sec:results}.

Now for an innocent user $\mathsf{I}_1 \neq \mathsf{I}_0$, we have
\begin{equation}
Pr[ 1 \mathsf{I}_0 \mathsf{I}_1 \in \mathcal{T}(y') \mid y',w_1,w_2,w_{\mathsf{I}_0},\mathsf{st} ] = p^{a_{\mathrm{HL}*\mathrm{L}}} (1-p)^{a_{\mathrm{LH}*\mathrm{H}}} \enspace.
\end{equation}
Therefore we have
\begin{equation}
Pr[ 1 \mathsf{I}_0 \mathsf{I}_1 \in \mathcal{T}(y') \mbox{ for some } \mathsf{I}_1 \in U_{\mathrm{I}} \mid y',w_1,w_2,w_{\mathsf{I}_0},\mathsf{st} ]
\leq (N-4) p^{a_{\mathrm{HL}*\mathrm{L}}} (1-p)^{a_{\mathrm{LH}*\mathrm{H}}}
\end{equation}
as there are $N-4$ choices of $\mathsf{I}_1$.
Similarly, we have
\begin{equation}
Pr[ 2 \mathsf{I}_0 \mathsf{I}_2 \in \mathcal{T}(y') \mbox{ for some } \mathsf{I}_2 \in U_{\mathrm{I}} \mid y',w_1,w_2,w_{\mathsf{I}_0},\mathsf{st} ]
\leq (N-4) p^{a_{\mathrm{H}*\mathrm{LL}}} (1-p)^{a_{\mathrm{L}*\mathrm{HH}}} \enspace.
\end{equation}
Hence the probability that $1 \mathsf{I}_0 \mathsf{I}_1,2 \mathsf{I}_0 \mathsf{I}_2 \in \mathcal{T}(y')$ for some $\mathsf{I}_1,\mathsf{I}_2 \in U_{\mathrm{I}}$, conditioned on the given $y'$, $w_1$, $w_2$, $w_{\mathsf{I}_0}$, and $\mathsf{st}$, is lower than the minimum of the two values $(N-4) p^{a_{\mathrm{HL}*\mathrm{L}}} (1-p)^{a_{\mathrm{LH}*\mathrm{H}}}$ and $(N-4) p^{a_{\mathrm{H}*\mathrm{LL}}} (1-p)^{a_{\mathrm{L}*\mathrm{HH}}}$, which is not higher than
\begin{equation}
\begin{split}
&\sqrt{ (N-4) p^{a_{\mathrm{HL}*\mathrm{L}}} (1-p)^{a_{\mathrm{LH}*\mathrm{H}}} \cdot (N-4) p^{a_{\mathrm{H}*\mathrm{LL}}} (1-p)^{a_{\mathrm{L}*\mathrm{HH}}} } \\
&= (N-4) \sqrt{p}^{\,a_{\mathrm{HL}*\mathrm{L}} + a_{\mathrm{H}*\mathrm{LL}}} \sqrt{1-p}^{\,a_{\mathrm{LH}*\mathrm{H}} + a_{\mathrm{L}*\mathrm{HH}}} \enspace.
\end{split}
\end{equation}

Now given $y'$, $w_1$, $w_2$, and $\mathsf{st}$, the probability that $w_{\mathsf{I}_0}$ attains the given values of $a_{\mathrm{HLLL}}$, $a_{\mathrm{HLHL}}$, $a_{\mathrm{HHLL}}$, $a_{\mathrm{LHLH}}$, $a_{\mathrm{LHHH}}$, and $a_{\mathrm{LLHH}}$ (denoted here by $\eta$) is the product of the following six values
\begin{equation}
\begin{split}
\binom{a_{\mathrm{HLL}*}}{a_{\mathrm{HLLL}}} (1-p)^{a_{\mathrm{HLLL}}}\allowbreak p^{a_{\mathrm{HLL}*} - a_{\mathrm{HLLL}}} \enspace,\enspace
\binom{a_{\mathrm{HLH}*}}{a_{\mathrm{HLHL}}} (1-p)^{a_{\mathrm{HLHL}}}p^{a_{\mathrm{HLH}*} - a_{\mathrm{HLHL}}} \enspace,\\
\binom{a_{\mathrm{HHL}*}}{a_{\mathrm{HHLL}}} (1-p)^{a_{\mathrm{HHLL}}}p^{a_{\mathrm{HHL}*} - a_{\mathrm{HHLL}}} \enspace,\enspace
\binom{a_{\mathrm{LHL}*}}{a_{\mathrm{LHLH}}} p^{a_{\mathrm{LHLH}}}(1-p)^{a_{\mathrm{LHL}*} - a_{\mathrm{LHLH}}} \enspace,\\
\binom{a_{\mathrm{LHH}*}}{a_{\mathrm{LHHH}}} p^{a_{\mathrm{LHHH}}}(1-p)^{a_{\mathrm{LHH}*} - a_{\mathrm{LHHH}}} \enspace,\enspace
\binom{a_{\mathrm{LLH}*}}{a_{\mathrm{LLHH}}} p^{a_{\mathrm{LLHH}}}(1-p)^{a_{\mathrm{LLH}*} - a_{\mathrm{LLHH}}} \enspace.
\end{split}
\end{equation}
By the above results, it follows that
\begin{equation}
\begin{split}
&Pr[ 1 \mathsf{I}_0 \mathsf{I}_1,2 \mathsf{I}_0 \mathsf{I}_2 \in \mathcal{T}(y') \mbox{ for some } \mathsf{I}_1,\mathsf{I}_2 \in U_{\mathrm{I}} \mid y',w_1,w_2,\mathsf{st} ] \\
&\leq \sum \eta (N-4) \sqrt{p}^{\,2 a_{\mathrm{HLLL}} + a_{\mathrm{HLHL}} + a_{\mathrm{HHLL}}} \sqrt{1-p}^{\,a_{\mathrm{LLHH}} + a_{\mathrm{LHLH}} + 2a_{\mathrm{LHHH}}} \enspace,
\end{split}
\end{equation}
where the sum runs over the possible values of $a_{\mathrm{HLLL}}$, $a_{\mathrm{HLHL}}$, $a_{\mathrm{HHLL}}$, $a_{\mathrm{LHLH}}$, $a_{\mathrm{LHHH}}$, and $a_{\mathrm{LLHH}}$.
Now by the above definition of $\eta$, the summand in the right-hand side is the product of $N-4$ and the following six values
\begin{equation}
\begin{split}
\binom{a_{\mathrm{HLL}*}}{a_{\mathrm{HLLL}}} (1-p)^{a_{\mathrm{HLLL}}}p^{a_{\mathrm{HLL}*}} \enspace,\enspace
\binom{a_{\mathrm{HLH}*}}{a_{\mathrm{HLHL}}} \left( (1-p)\sqrt{p} \right)^{a_{\mathrm{HLHL}}}p^{a_{\mathrm{HLH}*} - a_{\mathrm{HLHL}}} \enspace,\\
\binom{a_{\mathrm{HHL}*}}{a_{\mathrm{HHLL}}} \left( (1-p)\sqrt{p} \right)^{a_{\mathrm{HHLL}}}p^{a_{\mathrm{HHL}*} - a_{\mathrm{HHLL}}} \enspace,\enspace
\binom{a_{\mathrm{LHL}*}}{a_{\mathrm{LHLH}}} \left( p\sqrt{1-p} \right)^{a_{\mathrm{LHLH}}}(1-p)^{a_{\mathrm{LHL}*} - a_{\mathrm{LHLH}}} \enspace,\\
\binom{a_{\mathrm{LHH}*}}{a_{\mathrm{LHHH}}} p^{a_{\mathrm{LHHH}}}(1-p)^{a_{\mathrm{LHH}*}} \enspace,\enspace
\binom{a_{\mathrm{LLH}*}}{a_{\mathrm{LLHH}}} \left( p\sqrt{1-p} \right)^{a_{\mathrm{LLHH}}}(1-p)^{a_{\mathrm{LLH}*} - a_{\mathrm{LLHH}}} \enspace.
\end{split}
\end{equation}
Then by the binomial theorem, the sum is equal to
\begin{equation}
\label{eq:proof_type_III_bound_1}
\begin{split}
&(N-4) \left( p(2-p) \right)^{a_{\mathrm{HLL}*}} \left( p + (1-p)\sqrt{p} \right)^{a_{\mathrm{HLH}*} + a_{\mathrm{HHL}*}} \\
&\cdot \left( 1-p + p\sqrt{1-p} \right)^{a_{\mathrm{LHL}*} + a_{\mathrm{LLH}*}} \left( (1-p)(1+p) \right)^{a_{\mathrm{LHH}*}} \enspace.
\end{split}
\end{equation}

Given $y'$, $\mathsf{st}$, $w_1$, $w_2$, and $w_3$, we define, for $\alpha,\beta,\gamma,\delta \in \{\mathrm{H},\mathrm{L}\}$,
\begin{equation}
b_{\alpha \beta \gamma \delta} = |\{j \in [m] \mid y'_j = \xi^{\alpha}_j, w_{1,j} = \xi^{\beta}_j, w_{2,j} = \xi^{\gamma}_j, w_{3,j} = \xi^{\delta}_j\}| \enspace.
\end{equation}
Then by Marking Assumption, (\ref{eq:proof_type_III_bound_1}) is equal to
\begin{equation}
\label{eq:proof_type_III_bound_2}
\begin{split}
&(N-4) (2p - p^2)^{b_{\mathrm{HLLH}}}
\left( p + (1-p)\sqrt{p} \right)^{b_{\mathrm{HLHL}} + b_{\mathrm{HLHH}} + b_{\mathrm{HHLL}} + b_{\mathrm{HHLH}}} \\
&\cdot \left( 1-p + p\sqrt{1-p} \right)^{b_{\mathrm{LHLL}} + b_{\mathrm{LHLH}} + b_{\mathrm{LLHL}} + b_{\mathrm{LLHH}}}
(1-p^2)^{b_{\mathrm{LHHL}}} \\
&= (N-4) (2p - p^2)^{b_{\mathrm{HLLH}}}(1-p^2)^{b_{\mathrm{LHHL}}}
\left( p + (1-p)\sqrt{p} \right)^{b_{\mathrm{HLHL}} + b_{\mathrm{HHLH}}} \\
&\quad\cdot \left(1-p + p\sqrt{1-p} \right)^{b_{\mathrm{LHLH}} + b_{\mathrm{LLHL}}}
\left( p + (1-p)\sqrt{p} \right)^{b_{\mathrm{HLHH}} + b_{\mathrm{HHLL}}}
\left( 1-p + p\sqrt{1-p} \right)^{b_{\mathrm{LHLL}} + b_{\mathrm{LLHH}}} \enspace.
\end{split}
\end{equation}
By writing the right-hand side of (\ref{eq:proof_type_III_bound_2}) as $\eta'$, it follows that
\begin{equation}
\label{eq:proof_type_III_bound_3}
Pr[ 1 \mathsf{I}_0 \mathsf{I}_1,2 \mathsf{I}_0 \mathsf{I}_2 \in \mathcal{T}(y') \mbox{ for some } \mathsf{I}_1,\mathsf{I}_2 \in U_{\mathrm{I}} \mid w_1,w_2,w_3 ]
\leq \sum_{y',\mathsf{st}} Pr[y',\mathsf{st} \mid w_1,w_2,w_3] \eta' \enspace.
\end{equation}
Now we present the following key lemma, which will be proven later:
\begin{lemma}
\label{lem:proof_type_III_majority}
Among the possible pirate strategies $\rho$, the maximum value of the right-hand side of (\ref{eq:proof_type_III_bound_3}) is attained by majority vote attack $\rho_{\mathrm{maj}}$ (cf., Lemma \ref{lem:proof_type_II_majority}).
\end{lemma}
By (\ref{eq:proof_type_III_bound_3}), we have
\begin{equation}
\begin{split}
Pr[ 1 \mathsf{I}_0 \mathsf{I}_1,2 \mathsf{I}_0 \mathsf{I}_2 \in \mathcal{T}(y') \mbox{ for some } \mathsf{I}_1,\mathsf{I}_2 \in U_{\mathrm{I}} ]
&\leq \sum_{w_1,w_2,w_3} Pr[w_1,w_2,w_3] \sum_{y',\mathsf{st}} Pr[y',\mathsf{st} \mid w_1,w_2,w_3] \eta' \\
&= \sum_{y',\mathsf{st},w_1,w_2,w_3} Pr[y',\mathsf{st},w_1,w_2,w_3] \eta' \enspace.
\end{split}
\end{equation}
By virtue of Lemma \ref{lem:proof_type_III_majority}, the maximum value of the right-hand side is attained by majority vote attack $\rho_{\mathrm{maj}}$.
Now for $\rho = \rho_{\mathrm{maj}}$, the word $y'$ is uniquely determined by $w_1$, $w_2$, and $w_3$, and we have $b_{\mathrm{HLLH}} = b_{\mathrm{LHHL}} = b_{\mathrm{HLHL}} = b_{\mathrm{LHLH}} = b_{\mathrm{HHLL}} = b_{\mathrm{LLHH}} = 0$, $b_{\mathrm{HHLH}} = d_{\mathrm{HLH}}$, $b_{\mathrm{LLHL}} = d_{\mathrm{LHL}}$, $b_{\mathrm{HLHH}} = d_{\mathrm{LHH}}$, and $b_{\mathrm{LHLL}} = d_{\mathrm{HLL}}$, where, for $\alpha,\beta,\gamma \in \{\mathrm{H},\mathrm{L}\}$,
\begin{equation}
d_{\alpha \beta \gamma} = |\{j \in [m] \mid w_{1,j} = \xi^{\alpha}_j, w_{2,j} = \xi^{\beta}_j, w_{3,j} = \xi^{\gamma}_j\}| \enspace.
\end{equation}
This implies that
\begin{equation}
\eta'
= (N-4) \left( p + (1-p)\sqrt{p} \right)^{d_{\mathrm{LHH}} + d_{\mathrm{HLH}}}
\left( 1-p + p\sqrt{1-p} \right)^{d_{\mathrm{HLL}} + d_{\mathrm{LHL}}} \enspace.
\end{equation}
Put $d_{\mathrm{other}} = m - d_{\mathrm{HLL}} - d_{\mathrm{LHL}} - d_{\mathrm{LHH}} - d_{\mathrm{HLH}}$.
Now given $\mathsf{st}$, the probability that $w_1$, $w_2$ and $w_3$ attain the given values of $d_{\mathrm{HLL}}$, $d_{\mathrm{LHL}}$, $d_{\mathrm{LHH}}$ and $d_{\mathrm{HLH}}$ is
\begin{equation}
\binom{m}{d_{\mathrm{HLL}},d_{\mathrm{LHL}},d_{\mathrm{LHH}},d_{\mathrm{HLH}},d_{\mathrm{other}}} (p(1-p)^2)^{d_{\mathrm{HLL}}+d_{\mathrm{LHL}}}
(p^2(1-p))^{d_{\mathrm{LHH}}+d_{\mathrm{HLH}}} (1-2p(1-p))^{d_{\mathrm{other}}}
\end{equation}
which is independent of $\mathsf{st}$.
This implies that
\begin{equation}
\begin{split}
&\sum_{y',\mathsf{st},w_1,w_2,w_3} Pr[y',\mathsf{st},w_1,w_2,w_3] \eta' \\
&= \sum \binom{m}{d_{\mathrm{HLL}},d_{\mathrm{LHL}},d_{\mathrm{LHH}},d_{\mathrm{HLH}},d_{\mathrm{other}}} (N-4)
\left( p(1-p)^2(1-p + p\sqrt{1-p}) \right)^{d_{\mathrm{HLL}}+d_{\mathrm{LHL}}} \\
&\quad\cdot \left( p^2(1-p)(p + (1-p)\sqrt{p}) \right)^{d_{\mathrm{LHH}}+d_{\mathrm{HLH}}}
\left( 1-2p(1-p) \right)^{d_{\mathrm{other}}}
\end{split}
\end{equation}
(where the sum runs over the possible values of $d_{\mathrm{HLL}}$, $d_{\mathrm{LHL}}$, $d_{\mathrm{LHH}}$, and $d_{\mathrm{HLH}}$)
\begin{equation}
\begin{split}
&= \sum \binom{m}{d_{--\mathrm{L}},d_{--\mathrm{H}},d_{\mathrm{other}}} (N-4)
\left( p(1-p)^{5/2}(p + \sqrt{1-p}) \right)^{d_{--\mathrm{L}}} \\
&\quad\cdot \left( p^{5/2}(1-p)(1-p + \sqrt{p}) \right)^{d_{--\mathrm{H}}}
\left( 1-2p+2p^2 \right)^{d_{\mathrm{other}}}
\end{split}
\end{equation}
(where the sum runs over the possible values of $d_{--\mathrm{L}} = d_{\mathrm{HLL}} + d_{\mathrm{LHL}}$ and $d_{--\mathrm{H}} = d_{\mathrm{LHH}} + d_{\mathrm{HLH}}$)
\begin{equation}
= (N-4) \Bigl( p(1-p)^{5/2}(p + \sqrt{1-p})
+ p^{5/2}(1-p)(1-p + \sqrt{p}) + 1-2p+2p^2 \Bigr)^m
= (N-4) f_2(p)^m \enspace.
\end{equation}
By the above argument, the value $Pr[ 1 \mathsf{I}_0 \mathsf{I}_1,2 \mathsf{I}_0 \mathsf{I}_2 \in \mathcal{T}(y') \allowbreak \mbox{for some } \mathsf{I}_1,\mathsf{I}_2 \in U_{\mathrm{I}} ]$ for a general $\rho$ is also bounded by the above value.
Hence Proposition \ref{prop:probability_type_III} follows, by considering the number of choices of the pair $1,2$ and the innocent user $\mathsf{I}_0$.

To complete the proof of Proposition \ref{prop:probability_type_III}, we give a proof of Lemma \ref{lem:proof_type_III_majority}.
\begin{proof}
[Proof of Lemma \ref{lem:proof_type_III_majority}]
First, note that $1/2 \leq p < 1$, therefore $0 < 2p - p^2 < 1$, $0 < 1 - p^2 < 1$ and $0 < 1-p + p\sqrt{1-p} \leq p + (1-p) \sqrt{p} < 1$.
Now by the definition (\ref{eq:proof_type_III_bound_2}) of $\eta'$, for each $j \in [m]$ such that $w_{1,j} = w_{2,j} \neq w_{3,j}$, the value of $\eta'$ is increased by setting the $j$-th bit of the attack word $y$ to be $w_{1,j}$ instead of $w_{3,j}$ or `$?$' (which makes the values of $b_{\mathrm{HLLH}}$ and $b_{\mathrm{LHHL}}$ smaller).

We consider the case that $w_{1,j} = w_{3,j} \neq w_{2,j}$.
If $w_{1,j} = \xi^{\mathrm{H}}_j$, then the contribution of the $j$-th column to the value $\eta'$ is $p + (1-p)\sqrt{p}$ when $y'_j = w_{1,j}$ and $1-p + p\sqrt{1-p}$ when $y'_j = w_{2,j}$.
On the other hand, if $w_{1,j} = \xi^{\mathrm{L}}_j$, then the contribution of the $j$-th column to the value $\eta'$ is $1-p + p\sqrt{1-p}$ when $y'_j = w_{1,j}$ and $p + (1-p)\sqrt{p}$ when $y'_j = w_{2,j}$.
Recall the relation $1-p + p\sqrt{1-p} \leq p + (1-p) \sqrt{p}$.
Now the same argument as Lemma \ref{lem:proof_type_II_majority} implies that $Pr[w_{1,j} = \xi^{\mathrm{H}}_j] = p \geq 1-p = Pr[w_{1,j} = \xi^{\mathrm{L}}_j]$ in this case.
This implies that the value of the right-hand side of (\ref{eq:proof_type_III_bound_3}) is not decreased by setting $y'_j$ to be $w_{1,j}$ instead of $w_{2,j}$ (the detail of the proof is similar to the proof of Lemma \ref{lem:proof_type_II_majority}).
Similarly, in the case that $w_{1,j} \neq w_{2,j} = w_{3,j}$, the value of the right-hand side of (\ref{eq:proof_type_III_bound_3}) is not decreased by setting $y'_j$ to be $w_{2,j}$ instead of $w_{1,j}$.

Summarizing, the value of the right-hand side of (\ref{eq:proof_type_III_bound_3}) is not decreased by setting $y'_j$ to be the majority of $w_{1,j}$, $w_{2,j}$, and $w_{3,j}$, instead of the minority of them.
Hence the maximum value of the right-hand side of (\ref{eq:proof_type_III_bound_3}) is attained by the majority vote attack, concluding the proof of Lemma \ref{lem:proof_type_III_majority}.
\end{proof}

\subsection{Proof of Proposition \ref{prop:probability_type_IV}}
\label{subsec:proof_probability_type_IV}

To prove Proposition \ref{prop:probability_type_IV}, we fix an innocent user $\mathsf{I}$ and suppose that $S(i) < Z$ for every $i \in 123$.
Given $y'$, $w_1$, $w_2$, $w_3$, and $\mathsf{st}$, we define, for $\alpha,\beta,\gamma,\delta \in \{\mathrm{H},\mathrm{L}\}$,
\begin{equation}
a_{\alpha \beta \gamma \delta} = |\{j \in [m] \mid y'_j = \xi^{\alpha}_j, w_{1,j} = \xi^{\beta}_j, w_{2,j} = \xi^{\gamma}_j, w_{3,j} = \xi^{\delta}_j\}| \enspace.
\end{equation}
Then we have
\begin{equation}
\label{eq:proof_type_IV_bound_1}
Pr[ 12\mathsf{I},13\mathsf{I},23\mathsf{I} \in \mathcal{T}(y') \mid y',w_1,w_2,w_3,\mathsf{st} ]
= p^{a_{\mathrm{HLLH}} + a_{\mathrm{HLHL}} + a_{\mathrm{HHLL}}} (1-p)^{a_{\mathrm{LLHH}} + a_{\mathrm{LHLH}} + a_{\mathrm{LHHL}}} \enspace.
\end{equation}
Let $a_{\mathrm{L}}$ and $a_{\mathrm{H}}$ be as defined in Sect.~\ref{sec:results}.
For $x \in \{\mathrm{L},\mathrm{H}\}$, let $a^{\mathrm{u}}_x$ and $a^{\mathrm{d}}_x$ be the number of indices $j \in [m]$ of undetectable and detectable columns, respectively, such that $y'_j = \xi^{x}_j$.
Note that $a_{\mathrm{H}} = a^{\mathrm{u}}_{\mathrm{H}} + a^{\mathrm{d}}_{\mathrm{H}}$, while we have $a^{\mathrm{u}}_{\mathrm{H}} = a_{\mathrm{HHHH}}$ and $a^{\mathrm{u}}_{\mathrm{L}} = a_{\mathrm{LLLL}}$ by Marking Assumption.
Now we have
\begin{equation}
\begin{split}
&S(1) + S(2) + S(3) \\
&= \Bigl(3 a_{\mathrm{HHHH}} + 2(a_{\mathrm{HLHH}} + a_{\mathrm{HHLH}} + a_{\mathrm{HHHL}}) + a_{\mathrm{HLLH}} + a_{\mathrm{HLHL}} + a_{\mathrm{HHLL}} \Bigr) \log \frac{1}{p} \\
&\quad + \Bigl(3 a_{\mathrm{LLLL}} + 2(a_{\mathrm{LLLH}} + a_{\mathrm{LLHL}} + a_{\mathrm{LHLL}}) + a_{\mathrm{LLHH}} + a_{\mathrm{LHLH}} + a_{\mathrm{LHHL}} \Bigr) \log \frac{1}{1-p} \\
&= a^{\mathrm{u}}_{\mathrm{H}} \log \frac{1}{p} + a^{\mathrm{u}}_{\mathrm{L}} \log \frac{1}{1-p} + 2 \left( a_{\mathrm{H}} \log \frac{1}{p} + a_{\mathrm{L}} \log \frac{1}{1-p} \right) \\
&\quad - (a_{\mathrm{HLLH}} + a_{\mathrm{HLHL}} + a_{\mathrm{HHLL}}) \log \frac{1}{p} - (a_{\mathrm{LLHH}} + a_{\mathrm{LHLH}} + a_{\mathrm{LHHL}}) \log \frac{1}{1-p} \enspace,
\end{split}
\end{equation}
therefore
\begin{equation}
\begin{split}
&(a_{\mathrm{HLLH}} + a_{\mathrm{HLHL}} + a_{\mathrm{HHLL}}) \log\frac{1}{p}
+ (a_{\mathrm{LLHH}} + a_{\mathrm{LHLH}} + a_{\mathrm{LHHL}}) \log\frac{1}{1-p} \\
&= 2 \left( a_{\mathrm{H}} \log \frac{1}{p} + a_{\mathrm{L}} \log \frac{1}{1-p} \right) + a^{\mathrm{u}}_{\mathrm{H}} \log\frac{1}{p}
+ a^{\mathrm{u}}_{\mathrm{H}} \log\frac{1}{1-p} - S(1) - S(2) - S(3) \\
&> 2 \left( a_{\mathrm{H}} \log \frac{1}{p} + a_{\mathrm{L}} \log \frac{1}{1-p} \right) + a^{\mathrm{u}}_{\mathrm{H}} \log\frac{1}{p}
+ a^{\mathrm{u}}_{\mathrm{L}} \log\frac{1}{1-p} - 3 Z_0
\end{split}
\end{equation}
where we used the assumptions that $S(i) < Z$ for every $i \in 123$ and $Z \leq Z_0$.
By using the relation $a_{\mathrm{L}} = m - a_{\mathrm{H}}$ and the definition (\ref{eq:threshold_example}) of $Z_0$, the right-hand side of the above inequality is equal to
\begin{equation}
\begin{split}
&\ (3p-1) m \log\frac{1}{1-p}
+ a_{\mathrm{H}} \left( (2-3p) \log\frac{1}{p} + (1-3p) \log\frac{1}{1-p} \right) \\
&\quad + a^{\mathrm{u}}_{\mathrm{H}} \log\frac{1}{p} + a^{\mathrm{u}}_{\mathrm{L}} \log\frac{1}{1-p}
- 3 \sqrt{ \frac{1}{2} \left( \left( \log\frac{1}{p} \right)^2 a_{\mathrm{H}} + \left( \log\frac{1}{1-p} \right)^2 a_{\mathrm{L}} \right) \log \frac{N}{\varepsilon_0} } \\
&= (3p-1) m \log\frac{1}{1-p} + a^{\mathrm{u}}_{\mathrm{L}} \log\frac{1}{1-p}
+ a^{\mathrm{u}}_{\mathrm{H}} \left( (3-3p) \log\frac{1}{p} + (1-3p) \log\frac{1}{1-p} \right) \\
&\quad + a^{\mathrm{d}}_{\mathrm{H}} \left( (2-3p) \log\frac{1}{p} + (1-3p) \log\frac{1}{1-p} \right) \\
&\quad - 3 \Biggl( \frac{1}{2} \Biggl( \left( \log\frac{1}{1-p} \right)^2 m - \left( \left( \log\frac{1}{1-p} \right)^2 - \left( \log\frac{1}{p} \right)^2 \right) a_{\mathrm{H}} \Biggr) \log \frac{N}{\varepsilon_0} \Biggr)^{1/2}
\end{split}
\end{equation}
(where we used the relation $a_{\mathrm{H}} = a^{\mathrm{u}}_{\mathrm{H}} + a^{\mathrm{d}}_{\mathrm{H}}$)
\begin{equation}
\begin{split}
&\geq (3p-1) m \log\frac{1}{1-p}
+ a^{\mathrm{u}}_{\mathrm{H}} \left( (3-3p) \log\frac{1}{p} + (1-3p) \log\frac{1}{1-p} \right) \\
&\quad + a^{\mathrm{d}}_{\mathrm{H}} \left( (2-3p) \log\frac{1}{p} + (1-3p) \log\frac{1}{1-p} \right)
+ a^{\mathrm{u}}_{\mathrm{L}} \log\frac{1}{1-p} - 3 \sqrt{ \frac{1}{2} m \log \frac{N}{\varepsilon_0} } \log \frac{1}{1-p}
\end{split}
\end{equation}
(where we used the fact $\log(1/(1-p)) \geq \log(1/p) > 0$).
By applying the above inequalities to (\ref{eq:proof_type_IV_bound_1}), we have
\begin{equation}
\label{eq:proof_type_IV_bound_2}
\begin{split}
&Pr[ 12\mathsf{I},13\mathsf{I},23\mathsf{I} \in \mathcal{T}(y') \mid y',w_1,w_2,w_3,\mathsf{st} ] \\
&< (1-p)^{(3p-1)m} (1-p)^{-3 \sqrt{ (m/2) \log(N/\varepsilon_0) } }
\left( p^{3-3p} (1-p)^{1-3p} \right)^{a^{\mathrm{u}}_{\mathrm{H}}} (1-p)^{a^{\mathrm{u}}_{\mathrm{L}}} \left( p^{2-3p} (1-p)^{1-3p} \right)^{a^{\mathrm{d}}_{\mathrm{H}}} \enspace.
\end{split}
\end{equation}
We write the right-hand side of (\ref{eq:proof_type_IV_bound_2}) as $\eta$.
Then we have
\begin{equation}
\label{eq:proof_type_IV_bound_3}
\begin{split}
Pr[ 12\mathsf{I},13\mathsf{I},23\mathsf{I} \in \mathcal{T}(y') \mid w_1,w_2,w_3 ]
&< \sum_{\substack{y',\mathsf{st} \\ S(1),S(2),S(3) < Z}} Pr[y',\mathsf{st} \mid w_1,w_2,w_3 ] \eta \\
&\leq \sum_{y',\mathsf{st}} Pr[y',\mathsf{st} \mid w_1,w_2,w_3 ] \eta \enspace.
\end{split}
\end{equation}
Now we present the following key lemma, which will be proven later:
\begin{lemma}
\label{lem:proof_type_IV_majority}
Among the possible pirate strategies $\rho$, the maximum value of the right-hand side of (\ref{eq:proof_type_IV_bound_3}) is attained by majority vote attack $\rho_{\mathrm{maj}}$ (cf., Lemma \ref{lem:proof_type_II_majority}).
\end{lemma}
By (\ref{eq:proof_type_IV_bound_3}), we have
\begin{equation}
\begin{split}
Pr[ 12\mathsf{I},13\mathsf{I},23\mathsf{I} \in \mathcal{T}(y') ]
&< \sum_{w_1,w_2,w_3} Pr[w_1,w_2,w_3] \sum_{y',\mathsf{st}} Pr[y',\mathsf{st} \mid w_1,w_2,w_3] \eta \\
&= \sum_{y',\mathsf{st},w_1,w_2,w_3} Pr[y',\mathsf{st},w_1,w_2,w_3] \eta \enspace.
\end{split}
\end{equation}
By virtue of Lemma \ref{lem:proof_type_IV_majority}, the maximum value of the right-hand side is attained by majority vote attack $\rho_{\mathrm{maj}}$.
Now for $\rho = \rho_{\mathrm{maj}}$ and given $\mathsf{st}$, the probability that $w_1$, $w_2$, $w_3$ and $y'$ attain the given values of $a^{\mathrm{u}}_{\mathrm{H}}$, $a^{\mathrm{u}}_{\mathrm{L}}$, and $a^{\mathrm{d}}_{\mathrm{H}}$ is
\begin{equation}
\binom{m}{a^{\mathrm{u}}_{\mathrm{H}},a^{\mathrm{u}}_{\mathrm{L}},a^{\mathrm{d}}_{\mathrm{H}},a^{\mathrm{d}}_{\mathrm{L}}} (p^3)^{a^{\mathrm{u}}_{\mathrm{H}}} ((1-p)^3)^{a^{\mathrm{u}}_{\mathrm{L}}}
(3p^2(1-p))^{a^{\mathrm{d}}_{\mathrm{H}}} (3p(1-p)^2)^{a^{\mathrm{d}}_{\mathrm{L}}}
\end{equation}
which is independent of $\mathsf{st}$, where we put $a^{\mathrm{d}}_{\mathrm{L}} = m - a^{\mathrm{u}}_{\mathrm{H}} - a^{\mathrm{u}}_{\mathrm{L}} - a^{\mathrm{d}}_{\mathrm{H}}$.
Hence we have
\begin{equation}
\begin{split}
&\sum_{y',\mathsf{st},w_1,w_2,w_3} Pr[y',\mathsf{st},w_1,w_2,w_3] \eta \\
&= \sum_{a^{\mathrm{u}}_{\mathrm{H}},a^{\mathrm{u}}_{\mathrm{L}},a^{\mathrm{d}}_{\mathrm{H}}} \Biggl( \binom{m}{a^{\mathrm{u}}_{\mathrm{H}},a^{\mathrm{u}}_{\mathrm{L}},a^{\mathrm{d}}_{\mathrm{H}},a^{\mathrm{d}}_{\mathrm{L}}} (p^3)^{a^{\mathrm{u}}_{\mathrm{H}}} ((1-p)^3)^{a^{\mathrm{u}}_{\mathrm{L}}} (3p^2(1-p))^{a^{\mathrm{d}}_{\mathrm{H}}} (3p(1-p)^2)^{a^{\mathrm{d}}_{\mathrm{L}}} \eta \Biggr) \\
&= (1-p)^{(3p-1)m} (1-p)^{-3 \sqrt{ (m/2) \log(N/\varepsilon_{\mathrm{L}}) } } \\
&\cdot \sum \Biggl( \binom{m}{a^{\mathrm{u}}_{\mathrm{H}},a^{\mathrm{u}}_{\mathrm{L}},a^{\mathrm{d}}_{\mathrm{H}},a^{\mathrm{d}}_{\mathrm{L}}} \left(p^{6-3p} (1-p)^{1-3p}\right)^{a^{\mathrm{u}}_{\mathrm{H}}} \left((1-p)^4\right)^{a^{\mathrm{u}}_{\mathrm{L}}} \left(3p^{4-3p}(1-p)^{2-3p}\right)^{a^{\mathrm{d}}_{\mathrm{H}}} \left(3p(1-p)^2\right)^{a^{\mathrm{d}}_{\mathrm{L}}} \Biggr)
\end{split}
\end{equation}
(where the sum runs over the possible values of $a^{\mathrm{u}}_{\mathrm{H}}$, $a^{\mathrm{u}}_{\mathrm{L}}$, $a^{\mathrm{d}}_{\mathrm{H}}$, and $a^{\mathrm{d}}_{\mathrm{L}}$)
\begin{equation}
\begin{split}
&= (1-p)^{(3p-1)m} (1-p)^{-3 \sqrt{ (m/2) \log(N/\varepsilon_0) } } \\
&\quad\cdot \Bigl( p^{6-3p} (1-p)^{1-3p} + (1-p)^4 + 3p^{4-3p}(1-p)^{2-3p} + 3p(1-p)^2 \Bigr)^m \\
&= (1-p)^{(3p-1)m} (1-p)^{-3 \sqrt{ (m/2) \log(N/\varepsilon_0) } }
\Bigl( p^{4-3p}(p^2 - 3p + 3)(1-p)^{1-3p} + (1-p)^2(p^2 + p + 1) \Bigr)^m \\
&= (1-p)^{-3 \sqrt{ (m/2) \log(N/\varepsilon_0) } } f_3(p)^m \enspace.
\end{split}
\end{equation}
By the above argument, the value $Pr[ 12\mathsf{I},13\mathsf{I},23\mathsf{I} \in \mathcal{T}(y') ]$ for a general $\rho$ is also bounded by the above value.
Hence Proposition \ref{prop:probability_type_IV} follows, as there exist $N-3$ choices of the innocent user $\mathsf{I}$.

To complete the proof of Proposition \ref{prop:probability_type_IV}, we give a proof of Lemma \ref{lem:proof_type_IV_majority}.
\begin{proof}
[Proof of Lemma \ref{lem:proof_type_IV_majority}]
First note that, by Marking Assumption, the terms in $\eta$ other than $\left( p^{2-3p} (1-p)^{1-3p} \right)^{a^{\mathrm{d}}_{\mathrm{H}}}$ are independent of the choice of $y'$ for given $w_1$, $w_2$, and $w_3$.
An elementary analysis shows that $p^{2-3p} (1-p)^{1-3p}$ is an increasing function of $p \in [1/2,1)$, therefore $p^{2-3p} (1-p)^{1-3p} \geq (1/2)^{2-3/2}(1/2)^{1-3/2} = 1$.
Hence the value of $\eta$ will be increased by making the value of $a^{\mathrm{d}}_{\mathrm{H}}$ as large as possible.
By the same argument as Lemma \ref{lem:proof_type_II_majority}, under the condition that the $j$-th column is detectable, the probabilities that the majority among $w_{1,j}$, $w_{2,j}$, and $w_{3,j}$ is $\xi^{\mathrm{H}}_j$ and $\xi^{\mathrm{L}}_j$ are $p$ and $1-p$, respectively.
In other words, the probabilities that $\xi^{\mathrm{H}}_j$ is the majority and the minority among $w_{1,j}$, $w_{2,j}$, and $w_{3,j}$ are $p$ and $1-p$, respectively.
As $p \geq 1-p$, it follows that the value of the right-hand side of (\ref{eq:proof_type_IV_bound_3}) will not decrease by setting the $j$-th bit of $y'$ to be the majority of $w_{1,j}$, $w_{2,j}$, and $w_{3,j}$ instead of the minority of them (the detail of the proof is similar to the proof of Lemma \ref{lem:proof_type_II_majority}).
Hence the maximum value of the right-hand side of (\ref{eq:proof_type_IV_bound_3}) is attained by the majority vote attack, concluding the proof of Lemma \ref{lem:proof_type_IV_majority}.
\end{proof}

\subsection{Proof of Proposition \ref{prop:two_pirates_score}}
\label{subsec:proof_two_pirates_score}

First we introduce some notations.
Given the codewords $w_1$ and $w_2$ of the two pirates $1$ and $2$, let $a_{\mathrm{u}}$ and $a_{\mathrm{d}}$ denote the numbers of undetectable and detectable columns, respectively.
Then by Marking Assumption and the choice $p = 1/2$, we have $S(1) + S(2) = (2 a_{\mathrm{u}} + a_{\mathrm{d}}) \log 2$ regardless of the pirate strategy $\rho$.
This implies that, if $S(1) < Z$ and $S(2) < Z$, then we have
\begin{equation}
(2 a_{\mathrm{u}} + a_{\mathrm{d}}) \log 2 < 2 Z \leq 2 Z_0
= m \log 2 + \sqrt{ 2 m \log \frac{ N }{ \varepsilon_0 } } \log 2 \enspace.
\end{equation}
By the relation $a_{\mathrm{u}} + a_{\mathrm{d}} = m$, this implies that $2 m - a_{\mathrm{d}} < m + \sqrt{ 2 m \log(N/\varepsilon_0) }$, or equivalently $a_{\mathrm{d}} - m/2 > m/2 - \sqrt{ 2 m \log(N/\varepsilon_0) }$.
Now for each $j \in [m]$, the probability that the $j$-th column becomes detectable is $1/2$, therefore the expected value of $a_{\mathrm{d}}$ is $m/2$.
Then Hoeffding's Inequality (Theorem \ref{thm:Hoeffding}) implies that
\begin{equation}
\label{eq:proof_two_pirates_score_bound_1}
\begin{split}
Pr[ S(1) < Z \mbox{ and } S(2) < Z ]
&\leq Pr[ a_{\mathrm{d}} - m/2 > m/2 - \sqrt{ 2 m \log(N/\varepsilon_0) } ] \\
&\leq \exp\left( \frac{ -2 m^2 \left( m/2 - \sqrt{ 2 m \log(N/\varepsilon_0) } \right)^2 }{ m } \right) \\
&= \exp\left( \frac{ -m^2 \left( \sqrt{m} - \sqrt{ 8 \log(N/\varepsilon_0) } \right)^2 }{ 2 } \right)
\end{split}
\end{equation}
provided $m/2 - \sqrt{ 2m \log(N/\varepsilon_0) } > 0$.
The last condition is equivalent to that $m > 8 \log(N/\varepsilon_0)$ which is satisfied under the condition (\ref{eq:error_probability_condition_m}).
Now put $m = 8 \alpha \log(N/\varepsilon_0)$ with $\alpha > 1$.
Then under the condition (\ref{eq:error_probability_condition_m}), we have
\begin{equation}
\begin{split}
\frac{ m^2 \left( \sqrt{m} - \sqrt{8 \log (N/\varepsilon_0)} \right)^2 }{ 2 }
&= \frac{ m^2 }{ 2 } \left( \sqrt{\alpha} \cdot \sqrt{ 8 \log \frac{N}{\varepsilon_0} } - \sqrt{ 8 \log \frac{N}{\varepsilon_0} } \right)^2 \\
&= 4 m^2 \left( \sqrt{\alpha} - 1 \right)^2 \log\frac{N}{\varepsilon_0} \\
&> 16^2 \left( \log\frac{N}{\varepsilon_0} \right)^3 \left( 1 + \frac{ 1 }{ 16 \log(N/\varepsilon_0) } - 1 \right)^2
= \log\frac{N}{\varepsilon_0} \enspace,
\end{split}
\end{equation}
therefore the right-hand side of (\ref{eq:proof_two_pirates_score_bound_1}) is smaller than $\varepsilon_0/N$.
Hence the proof of Proposition \ref{prop:two_pirates_score} is concluded.

\subsection{Proof of Proposition \ref{prop:one_pirate_score}}
\label{subsec:proof_one_pirate_score}

Let $1 \in U$ be the unique pirate.
Then by Marking Assumption and the choice $p = 1/2$, we have $y' = w_1$ and $S(1) = m \log 2$, while $Z \leq Z_0 = (m/2) \log 2+ \sqrt{ (m/2) \log(N/\varepsilon_0) } \log 2$.
Now by the assumption $m \geq 2 \log(N/\varepsilon_0)$, we have
\begin{equation}
\frac{ S(1) - Z_0 }{ \log 2 }
= \frac{m}{2} - \sqrt{ \frac{m}{2} \log \frac{N}{\varepsilon_0} }
= \sqrt{ \frac{m}{2} } \left( \sqrt{ \frac{m}{2} } - \sqrt{ \log \frac{N}{\varepsilon_0} } \right) \geq 0 \enspace,
\end{equation}
therefore $S(1) \geq Z_0 \geq Z$.
Hence the proof of Proposition \ref{prop:one_pirate_score} is concluded.

\section{Conclusion}
\label{sec:conclusion}

In this article, we proposed a new construction of probabilistic $3$-secure codes and presented a theoretical evaluation of their error probabilities.
A characteristic of our tracing algorithm is to make use of both score comparison and search of the triples of \lq\lq parents'' for a given pirated fingerprint word.
Some numerical examples showed that code lengths of our proposed codes are significantly shorter than the previous provably secure $3$-secure codes.
Moreover, for the sake of improving efficiency of our tracing algorithm, we also proposed an implementation method for the algorithm, which seems indeed more efficient for an average case than the naive implementation.
A detailed evaluation of the proposed implementation method will be a future research topic.

\paragraph*{Acknowledgements.}
A preliminary version of this paper was presented at The 12th Information Hiding (IH 2010), Calgary, Canada, June 28--30, 2010 \cite{Nui10}.
The author would like to express his deep gratitude to Dr.\ Teddy Furon, who gave several invaluable comments and suggestions as the shepherd of the author's paper in that conference.
Also, the author would like to thank the anonymous referees at that conference for their precious comments.

\end{document}